\documentclass[11pt]{article}

\usepackage[letterpaper,margin=1in]{geometry}

\usepackage[T1]{fontenc}
\usepackage[utf8]{inputenc}

\usepackage{amsmath}
\usepackage{amssymb}
\usepackage{amsthm}

\IfFileExists{newtxtext.sty}{%
  \usepackage{newtxtext}%
  \usepackage{newtxmath}%
}{%
  \usepackage{mathptmx}%
}
\usepackage{bm}

\usepackage{mathtools}
\usepackage{microtype}

\usepackage{graphicx}
\graphicspath{{figures/}{./}}

\usepackage{algorithm}
\usepackage{algpseudocode}

\usepackage{tikz}
\usepackage{xcolor}

\newcommand{\ii}{\mathrm{i}}
\newcommand{\cH}{\mathcal H}
\newcommand{\cL}{\mathcal L}

\newcommand{\cP}{\mathcal P}

\newcommand{\cS}{\mathcal S}

\newcommand{\cR}{\mathcal R}

\newcommand{\Tr}{\operatorname{Tr}}

\newcommand{\wt}{\operatorname{wt}}
\newcommand{\cond}{\operatorname{cond}}
\newcommand{\Span}{\operatorname{span}}
\newcommand{\ip}[2]{\left\langle #1,#2\right\rangle}
\newcommand{\norm}[1]{\left\lVert #1\right\rVert}

\newcommand{\vecop}{\operatorname{vec}}
\def\argmax{\mathop{\rm arg\,max}}%
\def\argmin{\mathop{\rm arg\,min}}%

\usepackage{natbib}
\bibpunct[, ]{(}{)}{,}{a}{}{,}%
\usepackage[colorlinks=true,linkcolor=blue,citecolor=blue,urlcolor=blue]{hyperref}

\theoremstyle{plain}

\newtheorem{lemma}{Lemma}
\newtheorem{proposition}{Proposition}

\theoremstyle{definition}
\newtheorem{definition}{Definition}
\newtheorem{remark}{Remark}

\newenvironment{proof}[1]
  {\par\addvspace{6pt}\noindent\textit{#1}\enskip\ignorespaces}
  {\par\addvspace{6pt}}

\def\RUNAUTHOR#1{\gdef\theRUNAUTHOR{#1}}\def\theRUNAUTHOR{}
\def\RUNTITLE#1{\gdef\theRUNTITLE{#1}}\def\theRUNTITLE{}

\def\TITLE#1{\gdef\theTITLE{#1}}\def\theTITLE{}

\def\EMAIL#1{\texttt{#1}}
\newcommand{\authblock}{}
\def\AUTHOR#1{%
  \expandafter\gdef\expandafter\authblock\expandafter{\authblock
    \bigskip\par\begin{center}\large\textbf{#1}\end{center}}}
\def\AFF#1{%
  \expandafter\gdef\expandafter\authblock\expandafter{\authblock
    \par\begin{center}\small\itshape #1\end{center}}}
\long\def\ARTICLEAUTHORS#1{#1}

\long\def\ABSTRACT#1{\long\gdef\theABSTRACT{#1}}\def\theABSTRACT{}
\def\KEYWORDS#1{\gdef\theKEYWORDS{#1}}\def\theKEYWORDS{}

\def\theFUNDING{}
\def\theHISTORY{}

\makeatletter
\def\maketitle{%
  \begin{center}%
    {\LARGE\bfseries \theTITLE \par}%
  \end{center}%
  \authblock\par
  \bigskip
  \begin{abstract}%
    \noindent\theABSTRACT\par
    \ifx\theKEYWORDS\@empty\else
      \medskip\noindent\textbf{Key words:} \theKEYWORDS\par
    \fi
    \ifx\theFUNDING\@empty\else
      \medskip\noindent\textbf{Funding:} \theFUNDING\par
    \fi
    \ifx\theHISTORY\@empty\else
      \medskip\noindent\textbf{History:} \theHISTORY\par
    \fi
  \end{abstract}%
  \bigskip
}
\makeatother

\long\def\ACKNOWLEDGMENT#1{%
  \par\bigskip\noindent\textbf{Acknowledgments.}\enskip #1\par}

\newenvironment{APPENDIX}[1]
  {\par\appendix\section*{Appendix: #1}%
   \addcontentsline{toc}{section}{Appendix: #1}%
   \setcounter{section}{0}%
   }
  {\par}

\begin{document}

\RUNAUTHOR{Cipolla and Durastante}

\RUNTITLE{Pauli-Sparse Counterdiabatic Shortcuts for QAOA}

\TITLE{Pauli-Sparse regularised Counterdiabatic Shortcuts for Linear-Ramp QAOA}

\ARTICLEAUTHORS{%
\AUTHOR{Stefano Cipolla}
\AFF{School of Mathematical Sciences, University of Southampton, Southampton, UK, \EMAIL{s.cipolla@soton.ac.uk}}

\AUTHOR{Fabio Durastante}
\AFF{Department of Mathematics, University of Pisa, Pisa, Italy, \EMAIL{fabio.durastante@unipi.it}}
} %

\ABSTRACT{%
Combinatorial optimization is a leading target for quantum algorithms, but finite-depth QAOA
can suffer from strong diabatic errors when the interpolation Hamiltonian has small, or
exponentially small, spectral gaps.  We propose a Pauli-sparse counterdiabatic extension of
linear-ramp QAOA based on the regularised adiabatic gauge potential
\[
    \bigl(\mathcal L_H^2+\eta I\bigr)A_\lambda^{(\eta)}
    =
    -\mathrm{i}\mathcal L_H(\partial_\lambda H),
    \qquad
    \mathcal L_H(X)=[H,X].
\]
Instead of computing a dense AGP, we solve this equation approximately by an inexact
conjugate-gradient method in Pauli coordinates, truncating the Pauli expansion during the
iteration to obtain a gate-budget-aware set of implementable rotations.  The selected support is
then improved by a Galerkin refit and certified by an a posteriori residual bound. The regularization parameter \(\eta\) acts as an energy-resolution scale: it suppresses
transitions below \(\sqrt{\eta}\) while retaining larger-gap transitions.  Thus, the method can
avoid resolving exponentially small splittings inside a low-energy solution manifold while
reducing leakage away from it.  Numerical experiments on Ferromagnetic Chain (FC) and perturbed
FC--MaxCut/MarketSplit instances show that the resulting LR-CD-QAOA ansatz improves approximation
ratios over the uncorrected linear ramp, especially in regimes where LR-QAOA remains far from
the optimum. Overall, the proposed regularized LR-CD-QAOA framework substantially broadens the practical
applicability of QAOA to QUBO optimization by improving its robustness across heterogeneous
problem landscapes, including instances with near-degenerate low-energy structures and small
spectral gaps.
}%


\maketitle

\section{Introduction}%

Combinatorial optimization is one of the principal application areas for quantum optimization, see \cite{Abbas2024QuantumOptimization}.
A standard modelling route is to encode a discrete objective function into an Ising or QUBO
Hamiltonian \(H_C\), so that low-energy quantum states represent high-quality classical
solutions \cite{Lucas2014,giovagnoli2025introductionquantumapproximateoptimization}.
This Hamiltonian viewpoint underlies both adiabatic quantum optimization and gate-model
variational algorithms.  In particular, the Quantum Approximate Optimization Algorithm
(QAOA) prepares a parameterized state by alternating evolutions under a cost Hamiltonian
and a mixer Hamiltonian \cite{FarhiGoldstoneGutmann2014,Hadfield2019,Zhou2020}.  QAOA
and its variants have consequently become a central testbed for understanding whether
near-term quantum processors can produce useful optimization performance on structured
combinatorial problems \cite{LotshawHumbleHerrmanOstrowskiSiopsis2021,
WilkieGaidaiOstrowskiHerrman2024,montanezbarrera2025evaluatingperformancequantumprocessing,PhysRevApplied.22.064068}.

A major practical obstacle is the classical optimization of the circuit parameters.  At depth
\(p\), standard QAOA contains \(2p\) angles, and the resulting training landscape is nonconvex.
This has motivated parameter-reduction strategies, including linear-ramp QAOA (LR-QAOA),
where the layer angles are constrained to follow an adiabatic-inspired schedule depending on
only a small number of amplitudes \cite{Dehn2025,McDowall2026}.  Such schedules reduce the
outer-loop optimization burden and are particularly attractive when the cost of parameter
training dominates the quantum runtime.  However, they inherit a central weakness of finite-time
adiabatic evolution: when the interpolating Hamiltonian
$
    H(\lambda)=(1-\lambda)H_B+\lambda H_C
$
has small spectral gaps, and especially exponentially small bottleneck gaps, a short ramp may
generate strong diabatic transitions away from the desired low-energy sector
\cite{arezzo2026continuoustimequantumcontrolexponentially}.  Hence, although LR-QAOA is
parameter efficient, it may be dynamically inefficient on hard Hamiltonians.

Counterdiabatic driving provides a principled way to address this problem.  In the ideal
continuous-time setting, the adiabatic gauge potential (AGP) generates a correction that cancels
diabatic transitions induced by changes in \(\lambda\)
\cite{Kolodrubetz2017,SelsPolkovnikov2017,Hatomura_2024}.  This idea has inspired several
counterdiabatic and shortcut-to-adiabaticity approaches to quantum optimization and QAOA
\cite{WurtzLove2021,PhysRevResearch.4.L042030,TakahashiDelCampo2024,
MorawetzPolkovnikov2025,wbbs-s8fs,tang2026weightednestedcommutatorsscalable,pqhl-nbtk, PRXQuantum.5.020362}.
These works show that counterdiabatic corrections can be constructed from local variational
ansatz, commutator expansions, Krylov-space descriptions, or performance-guaranteed
control principles. In particular, regularised AGP viewpoint is not new in itself, see \cite{MorawetzPolkovnikov2025}.  The open
algorithmic question addressed here is different: how can one compute a Hamiltonian-adapted
regularised AGP for an optimization problem, while keeping only the Pauli terms that are useful
under a finite gate budget?

This paper proposes an answer based on an inexact Pauli-coordinate conjugate-gradient
construction.  We consider the regularised AGP equation
\begin{equation}
    \bigl(\mathcal L_{H(\lambda)}^2+\eta I\bigr)A_\lambda^{(\eta)}
    =
    -\mathrm{i}\,\mathcal L_{H(\lambda)}(G),
    \qquad
    G=\partial_\lambda H(\lambda)=H_C-H_B,
    \qquad
    \mathcal L_H(X)=[H,X].
    \label{eq:intro-regularised-agp}
\end{equation}
The operator \(\mathcal L_{H(\lambda)}^2+\eta I\) is Hermitian positive definite for
\(\eta>0\), so the AGP equation can be treated as a large linear system in Hilbert--Schmidt
space.  Instead of forming the dense \(4^n\)-dimensional Liouvillian matrix, we apply this
operator directly in the Pauli basis by symbolic commutators.  The conjugate-gradient
iteration is then deliberately made inexact by truncating the Pauli expansion after each step.
Thus the method does not first compute a dense AGP and then compress it; rather, it discovers
and refines a sparse set of Pauli generators during the iterative solve itself.

This distinction is important for optimization on quantum machines.  The regularization
parameter \(\eta\) acts as an energy-resolution scale: transitions with energy differences much
smaller than \(\sqrt{\eta}\) are suppressed, while larger transitions are retained.  Therefore, for
Hamiltonians with exponentially small splittings inside a target low-energy manifold, the
regularised AGP need not resolve every microscopic gap.  Instead, it can be tuned to suppress
leakage out of the useful low-energy sector.  This provides a mechanism for improving
finite-depth ramp dynamics even when exact adiabatic following would be obstructed by
exponentially small gaps.

The approach also connects naturally with the operations-research perspective on quantum
optimization.  Recent work has used quantum subroutines inside classical optimization
frameworks such as interior-point methods \cite{ApersGribling2023}, and optimization-based
methods have been developed for quantum circuit synthesis and compilation
\cite{NagarajanLockwoodCoffrin2021,HendersonNagarajanCoffrin2023}.  In contrast, the
present work uses classical sparse linear algebra and projection ideas to construct a quantum
optimization circuit.  The linear-algebra backbone is an inexact Krylov method; see, e.g.,
\cite{MR1990645,MR2058070,MR1740397}, while the Pauli algebra and conditioning arguments
rely on standard matrix-analysis facts from \cite{HornJohnson}.  The resulting algorithm is a hybrid
optimization method in a precise sense: a classical truncated Krylov solver selects the
counterdiabatic Pauli rotations, and the quantum circuit uses these rotations to improve the
optimization dynamics.
\subsection{Contributions, relation with existing work and paper organization}
\label{sec:contributions}

The contribution of this paper is a Pauli-sparse, regularised counterdiabatic extension of
LR-QAOA based on an inexact conjugate-gradient solution of the regularised AGP equation.
The method is designed for combinatorial optimization Hamiltonians whose low-energy
structure is useful even when internal spectral gaps are very small.  The paper is organized
around three main contributions. 

First, unlike existing CD-QAOA and shortcut-to-adiabaticity methods that typically prescribe
the counterdiabatic ansatz in advance through local AGP terms, nested commutators, weighted
commutators, exact or truncated commutator-generated AGP bases, or Krylov-space control
operators
\cite{SelsPolkovnikov2017,WurtzLove2021,PhysRevResearch.4.L042030,
TakahashiDelCampo2024,MorawetzPolkovnikov2025,SciPostPhys.18.1.014,
wbbs-s8fs,tang2026weightednestedcommutatorsscalable,pqhl-nbtk}, we start from the
regularised AGP equation \eqref{eq:intro-regularised-agp} and solve it directly in Pauli
coordinates.  The conjugate-gradient iteration is deliberately inexact: after each Krylov step
the Pauli expansion is truncated, so that the algorithm simultaneously approximates the AGP
and selects a gate-budget-aware set of implementable Pauli rotations.
Sections~\ref{sec:notation} and~\ref{sec:regularised-agp} introduce the Hamiltonian notation,
the regularised AGP equation, and the energy-filter interpretation that underpins this
construction.

Second, although regularised and approximate AGPs are already well established in
counterdiabatic driving \cite{Kolodrubetz2017,SelsPolkovnikov2017,Hatomura_2024}, we use
regularisation as an energy-resolution mechanism for discrete optimization.  In hard
Hamiltonians, exact adiabaticity may be limited by exponentially small gaps
\cite{arezzo2026continuoustimequantumcontrolexponentially}.  The regularised AGP suppresses
transitions below the scale \(\sqrt{\eta}\), allowing the method to avoid resolving exponentially
small splittings inside a low-energy manifold while still reducing leakage out of it.  This
differs from schedule-only LR-QAOA approaches \cite{Dehn2025,McDowall2026}, which modify
the ramp parameters but do not construct a Hamiltonian-adapted counterdiabatic operator.
Section \ref{subsec:residual_certificate} develops the projected formulation and the residual certificate that
make this energy-resolution viewpoint computable.

Third, the paper develops a sparse linear-algebra pipeline for producing an implementable CD
layer.  Instead of forming the \(4^n\)-dimensional Liouvillian matrix, we apply
\(\mathcal L_H^2+\eta I\) by symbolic Pauli commutators and use truncated CG as a
support-discovery method, importing ideas from inexact Krylov methods
\cite{MR1990645,MR2058070,MR1740397}.  The discovered support is then improved by a
Galerkin refit, separating Pauli-term selection from coefficient estimation, in contrast with
heuristic modifications of the phase operator, mixer, or schedule
\cite{Hadfield2019,Zhou2020,LotshawHumbleHerrmanOstrowskiSiopsis2021,
WilkieGaidaiOstrowskiHerrman2024}.  This is related in spirit to optimization-based quantum
circuit design \cite{NagarajanLockwoodCoffrin2021,HendersonNagarajanCoffrin2023}, but here
the optimization problem being solved is the regularised AGP linear system rather than a direct
circuit-synthesis problem.  Section \ref{sec:pauli_cg_agp} presents the truncated
Pauli-CG method, the Galerkin refit, and the a posteriori bound
\[
    \|A-A_\lambda^{(\eta)}\|_{\mathrm{HS}}
    \leq
    \frac{\|R\|_{\mathrm{HS}}}{\eta},
    \qquad
    R=b-\bigl(\mathcal L_H^2+\eta I\bigr)A .
\]
Section~\ref{sec:qaoa_implementation} then inserts the computed Pauli-sparse generator into
a low-dimensional LR-CD-QAOA circuit, where only
\((\Delta_\beta,\Delta_\gamma,\Delta_{\mathrm{CD}})\) are optimized rather than all \(2p\)
standard QAOA angles \cite{FarhiGoldstoneGutmann2014,Zhou2020,
giovagnoli2025introductionquantumapproximateoptimization}.

Finally, Section~\ref{sec:numerics} validates the method on Ferromagnetic Chain (FC) and
perturbed FC--MaxCut Hamiltonians.  These experiments show that plain LR-QAOA can remain
separated from the optimum when the small-gap FC structure dominates, whereas the
truncated Pauli-CG AGP correction substantially improves the approximation ratio by
suppressing diabatic leakage away from the low-energy solution manifold.  

Overall, the proposed regularized LR-CD-QAOA framework here proposed, substantially broadens the practical
applicability of QAOA to QUBO optimization by improving its robustness across heterogeneous
problem landscapes, including instances with near-degenerate low-energy structures and small
spectral gaps.

To conclude this section, we note that, for readability, all proofs of the results presented here are provided in the Appendix.

\subsection{Notation and problem setting} \label{sec:notation}

Let $n$ be the number of qubits and let $\cH=(\mathbb C^2)^{\otimes n}$.  For a positive Hermitian weight $\rho(H)$ satisfying $[\rho(H),H]=0$, define
\begin{equation*}
    (X,Y)_\rho
    :=\frac12\Tr\!\left[\rho(H)\left(X^\dagger Y+Y X^\dagger\right)\right],
\end{equation*}
where for $A \in \cH$, $A^{\dagger} = \bar{A}^T$. We also use {Hilbert–Schmidt inner product}
\begin{equation}
    \rho(H)=2^{-n}I,
    \qquad
    \ip{X}{Y}_{HS}:=(X,Y)_\rho=2^{-n}\Tr(X^\dagger Y),
    \qquad
    \norm{X}_{HS}^2=\ip{X}{X}{_{HS}}.
    \label{eq:hs-inner-product}
\end{equation}
The Pauli strings are
\begin{equation}\label{eq:pauli_strings}
 \mathcal{P}_n:= \{  P=P_1\otimes\cdots\otimes P_n \}, \hbox{  s.t. } 
    \qquad P_j\in\{I,X,Y,Z\},
\end{equation}
and satisfy $\ip{P}{Q}=\delta_{PQ}$.  The weight $\wt(P)$ is the number of non-identity single-qubit factors in $P$.

We consider the Hamiltonian $H(\lambda)$ acting on $\cH$ obtained \textit{annealing} the Hermitian  Hamiltonians  $H_B$ (\textit{mixer}) and $H_C$ (\textit{cost}),
    $H(\lambda)=(1-\lambda)H_B+\lambda H_C,
    \qquad \lambda\in[0,1]$, with derivative $G:=\partial_\lambda H(\lambda)=H_C-H_B$.
In numerical implementations, after fixing the
 computational basis, \(H_B\) and \(H_C\) are represented by a \(2^n\times 2^n\) Hermitian matrix. For a fixed Hermitian $H$, define the commutator superoperator, see~\cite{TakahashiDelCampo2024}, as
$
    \cL_H(X):=[H,X]=HX-XH.
$
We use the column-stacking vectorization operator $\vecop(\cdot)$.  For matrices $A,X,B$ of compatible
sizes,
$
    \vecop(AXB)=(B^T\otimes A)\vecop(X).
$
Therefore, the commutator superoperator has the Kronecker representation
\begin{equation}
    \vecop(\cL_H(X))=L_H\vecop(X),
    \qquad
    L_H:=I_N\otimes H-H^T\otimes I_N,
    \qquad N=2^n.
    \label{eq:liouvillian-kron}
\end{equation}
The square of the Liouvillian is represented hence, by
$
    \vecop(\cL_H^2(X))=L_H^2\vecop(X),
$
where
\begin{equation}
    L_H^2
    =I_N\otimes H^2-2H^T\otimes H+(H^T)^2\otimes I_N.
    \label{eq:expanded-kron-square}
\end{equation}
For Hermitian $H$, the matrix $L_H$ is Hermitian with respect to the Euclidean
inner product induced by \eqref{eq:hs-inner-product}.  Hence $L_H^2$ is
Hermitian positive semidefinite. We denote the 2-norm condition number by $\cond(\cdot)$. For Hermitian matrices, $X \succeq Y$ ($X \succ Y$) means $X - Y$ is positive semidefinite (definite).
\section{Regularised adiabatic gauge potential in commutator form} \label{sec:regularised-agp}

The exact adiabatic gauge potential $A_\lambda$ is the Hermitian operator that cancels diabatic transitions generated by the path $H(\lambda)$.  In finite dimension it may be characterized, after fixing the diagonal gauge, by the commutator equation, see \cite[Eq. (6)]{TakahashiDelCampo2024}
\begin{equation}
    \cL_{H(\lambda)}^2(A_\lambda)=-\ii\,\cL_{H(\lambda)}(G).
    \label{eq:unregularised-agp}
\end{equation}
Using the Kronecker formulation in~\eqref{eq:liouvillian-kron}, this is equivalently the $N^2$-dimensional linear system
\begin{equation}
    L_{H(\lambda)}^2\,\vecop(A_\lambda)
    =-\ii L_{H(\lambda)}\,\vecop(G).
    \label{eq:unregularised-agp-kron}
\end{equation}
The unregularised system is singular because $L^2_H$ has a nontrivial nullspace,
corresponding to operators that commute with $H$.   We therefore use the regularised Hermitian equation
\begin{equation}
   {
    \left(\cL_{H(\lambda)}^2+\eta I\right)A_\lambda^{(\eta)}
    =-\ii\,\cL_{H(\lambda)}(G),
    \qquad \eta>0,
    }
    \label{eq:regularised-agp}
\end{equation}
see \cite{MorawetzPolkovnikov2025}.  Equivalently,
$
    [H,[H,A_\lambda^{(\eta)}]]+\eta A_\lambda^{(\eta)}
    =-\ii[H,G]
$.
In Kronecker form, \eqref{eq:regularised-agp} becomes
\begin{equation}
   {
    \left(L_{H(\lambda)}^2+\eta I_{N^2}\right)\vecop(A_\lambda^{(\eta)})
    =-\ii L_{H(\lambda)}\vecop(G).
    }
    \label{eq:regularised-agp-kron}
\end{equation}
By expanding the square as in~\eqref{eq:expanded-kron-square}, the coefficient matrix is explicitly given by
\begin{equation*}
\begin{aligned}
    L_H^2+\eta I_{N^2}
    & = I_N\otimes H^2-2H^T\otimes H+(H^T)^2\otimes I_N+
    \eta I_{N^2}.
\end{aligned}
\end{equation*}
It is important to note also that the regularised equation can be interpreted as the normal equation associated with the regularised variation formulation
\begin{equation}
    \Phi_\eta(A)
    =\frac12\norm{G-\ii\cL_H(A)}^2_{HS}+\frac{\eta}{2}\norm{A}_{HS}^2= \frac{1}{2}\ip{A }{(\cL_H^2+\eta I)(A) }_{HS}  {+ \ip{\ii\cL_H(G)}{A}_{HS}} +C   ,
    \label{eq:tikhonov-functional}
\end{equation}
up to gauge-invariant diagonal terms, see \cite{SelsPolkovnikov2017}.  In vectorized form,
\begin{equation}
    \Phi_\eta(a)
    =\frac12\norm{g-\ii L_Ha}_2^2+\frac{\eta}{2}\norm{a}_2^2,
    \qquad a=\vecop(A),\quad g=\vecop(G).
    \label{eq:tikhonov-functional-kron}
\end{equation}
Since $L_H=L_H^\dagger$ for Hermitian $H$, stationarity of
\eqref{eq:tikhonov-functional-kron} gives
$(L_H^2+\eta I_{N^2})a=-\ii L_Hg$,
which is \eqref{eq:regularised-agp-kron}.  The right-hand side $-\ii[H,G]$ is
Hermitian whenever $H$ and $G$ are Hermitian, hence the solution of
\eqref{eq:regularised-agp} is also Hermitian.

\subsection{Spectral representation and filter interpretation}

To better understand the effect of filtering on the spectral components, we 
begin by expanding the Hamiltonian $H$ in an eigenbasis where it assumes a 
diagonal form. Let $H = \sum_{a=1}^N E_a |a\rangle\langle a|$, $N=2^n$,
where $\{|a\rangle\}_{a=1}^N \subset \mathbb{C}^N$ forms an orthonormal basis such that 
$\langle a | b \rangle = \delta_{ab}$, with $\delta_{ab}$ denoting the Kronecker delta. 
Furthermore, we define the energy differences and the matrix elements of an operator $G$ as
$\omega_{ab} = E_a - E_b$, for  $G_{ab} = \langle a|G|b\rangle$.
In this basis, the diagonal gauge choice implies $(A_\lambda)_{aa} = 0$.
\begin{lemma}\label{lem:filter-formula}
For $a\neq b$, the unregularised and regularised gauge potentials satisfy
\begin{equation}
    (A_\lambda)_{ab}=-\frac{\ii G_{ab}}{\omega_{ab}},
    \qquad
    (A_\lambda^{(\eta)})_{ab}
    =-\frac{\ii\omega_{ab}}{\omega_{ab}^2+\eta}G_{ab}.
    \label{eq:filter-formula}
\end{equation}
Thus
$
    (A_\lambda^{(\eta)})_{ab}
    =f_\eta(\omega_{ab})(A_\lambda)_{ab}$,
    and $
    f_\eta(\omega)=\frac{\omega^2}{\omega^2+\eta}.
$
\end{lemma}

As noted in~\cite{MorawetzPolkovnikov2025}, this regularization acts as a 
high-pass filter: transitions with $|\omega| \gg \sqrt{\eta}$ remain virtually 
unchanged, whereas transitions with $|\omega| \ll \sqrt{\eta}$ are suppressed.

\subsection{Distance to the unregularised gauge potential}

The filtering of the eigenvalues of $H$ introduced in Lemma~\ref{lem:filter-formula} 
induces a discrepancy between the original and regularized Hamiltonians. 
We quantify the resulting error in Proposition~\ref{prop:distance} by 
substituting the regularized operator into the original expression.

\begin{proposition}[Distance between $A_\lambda^{(\eta)}$ and $A_\lambda$]
\label{prop:distance}
Assume that 
$|\omega_{ab}|\geq \Delta>0$. Then
\begin{equation}
    \norm{A_\lambda^{(\eta)}-A_\lambda}_{HS}^2
    =\sum_{a\neq b}
    \frac{\eta^2}{\omega_{ab}^2(\omega_{ab}^2+\eta)^2}|G_{ab}|^2,
    \label{eq:distance-exact}
\end{equation}
and therefore
\begin{equation*}
   {
    \norm{A_\lambda^{(\eta)}-A_\lambda}_{HS}
    \leq
    \frac{\eta}{\Delta(\Delta^2+\eta)}\norm{G_{\rm off}}_{HS} {\text{ where } G_{\rm off} = \sum_{a \neq b} | G_{ab} |^2}.
    }
\end{equation*}
In particular, for fixed gap $\Delta$, the distance is $O(\eta)$ as $\eta\downarrow0$.
\end{proposition}

Proposition~\ref{prop:distance} quantifies the error introduced by regularizing the adiabatic gauge potential. The exact Hilbert--Schmidt distance in~\eqref{eq:distance-exact} shows that $\eta$ suppresses contributions from small energy gaps $\omega_{ab}$. For a spectrum gapped by $\Delta$, the upper bound is controlled by the off-diagonal elements of $G$. Crucially, as $\eta \downarrow 0$, the distance vanishes as $O(\eta)$, ensuring the unregularized potential is faithfully recovered away from degeneracies. However, this bound diverges as $\Delta \to 0$. To show that the regularization handles these singularities globally, Proposition~\ref{prop:boundedness} establishes a uniform upper bound on the norm of the potential that remains valid even in the gapless limit.

\begin{proposition}
\label{prop:boundedness}
For every $\eta>0$,
$
   {
    \norm{A_\lambda^{(\eta)}}_{HS}
    \leq
    \frac{1}{2\sqrt\eta}\norm{G_{\rm off}}_{HS}.
    }
$
\end{proposition}

\subsection{Well-posedness and conditioning}

In addition to ensuring analytical boundedness, the regularization provides numerical stability. Computing the adiabatic gauge potential via~\eqref{eq:regularised-agp-kron} requires inverting a Liouvillian system that is ill-conditioned near degeneracies. Proposition~\ref{prop:conditioning} shows how the parameter $\eta$ bounds this condition number both globally and on the off-diagonal subspace.

\begin{proposition}
\label{prop:conditioning}
Let $L_H=I_N\otimes H-H^T\otimes I_N$ be the Kronecker matrix representation
of $\cL_H$.  Let $\Omega=\norm{L_H}_2=\norm{\cL_H}_2$.  The regularised matrix
$
    L_H^2+\eta I_{N^2}
$
is Hermitian positive definite and satisfies
\begin{equation}
    \cond(L_H^2+\eta I_{N^2})
    \leq
    \frac{\Omega^2+\eta}{\eta}.
    \label{eq:full-conditioning}
\end{equation}
If we restrict to the off-diagonal subspace spanned by matrix units
$|a\rangle\langle b|$, $a\neq b$, on which all nonzero gaps satisfy
$|\omega_{ab}|\geq\Delta$, then
\begin{equation}
    \cond\!\left((L_H^2+\eta I_{N^2})|_{\rm off}\right)
    \leq
    \frac{\Omega^2+\eta}{\Delta^2+\eta}.
    \label{eq:gauge-fixed-conditioning}
\end{equation}
Furthermore,
\begin{equation}
    \Omega\leq 2\norm{H}_2,
    \qquad
    \cond(L_H^2+\eta I_{N^2})
    \leq 1+\frac{4\norm{H}_2^2}{\eta}.
    \label{eq:conditioning-h-bound}
\end{equation}
\end{proposition}

\begin{remark}
The role of $\eta$ is twofold.  First, it suppresses small-gap components that would otherwise produce large coefficients in the exact gauge potential.  Second, it gives a uniformly well-posed projected linear system.  Hence regularization is not merely a numerical trick: it also acts as a circuit-amplitude and rotation-budget control mechanism.
\end{remark}

\begin{definition}
{Throughout} the paper, we distinguish between the operator-level regularised Liouvillian
$\mathcal B_\eta := \cL_H^2+\eta I$,
acting on operators equipped with the Hilbert--Schmidt inner product, and its Kronecker
matrix representation $B_\eta := L_H^2+\eta I_{N^2}$.
When no ambiguity is possible, statements about \(\mathcal B_\eta\) and \(B_\eta\) are identified
through the vectorization isomorphism.
\end{definition}

\begin{lemma}
\label{lem:spectral_floor}
$L_H^2+\eta I_{N^2}\succeq\eta I_{N^2}$, with
$\lambda_{\min}(B_\eta)=\eta$ on the full Pauli space and
$\lambda_{\min}(B_\eta|_{\mathcal V})\ge\eta$ on any restricted subspace
$\mathcal V$ {, where $B_\eta|_{\mathcal{V}} = V^\dagger B_\eta V$ for $V$ an orthonormal basis of $\mathcal{V}$}. In particular $\norm{B_\eta^{-1}}_2\le\eta^{-1}$.
\end{lemma}

\section{Galerkin projections and a posteriori residual certificate}
\label{subsec:residual_certificate}

{Let \(\mathcal P_n\) denote the set of \(n\)-qubit Hermitian Pauli strings~\eqref{eq:pauli_strings}.} For a Hermitian operator $X$ we write
$\operatorname{supp}(X):=\{P\in\cP_n:\langle P,X\rangle_{\rm HS}\neq 0\}$ for
its Pauli support.

\begin{definition}
\label{def:restricted_master}
For a fixed support $\cS=\{P_1,\dots,P_\ell\}\subset\cP_n$, the
\emph{restricted problem} is
\begin{equation} \label{eq:restricted_problem}
A_{\cS}^{(\eta)} :=
\argmin\bigl\{\Phi_\eta(A)\;:\;A\in\Span\cS\bigr\},
\end{equation}
with $\Phi_\eta$ being the regularised functional in \eqref{eq:tikhonov-functional}. In vectorized form, setting $A_{\cS}^{(\eta)}=Q_{\cS}a^{(\eta)}$, we obtain the Galerkin projected  equation
$	Q_{\cS}^\dagger(L_H^2+\eta I)Q_{\cS}\,a^{(\eta)}
	=
	-\ii\,Q_{\cS}^\dagger L_H\vecop(G)$,
see \cite[Chap. 5]{MR1990645}. The \emph{residual equation} is then
\begin{equation}
	\cR_{\cS}
	:=
	-\ii\,\cL_H(G)-\cL_H^2(A_{\cS}^{(\eta)})-\eta A_{\cS}^{(\eta)}
	=
	\operatorname{unvec}\bigl(b-B_\eta A_{\cS}^{(\eta)}\bigr),
	\label{eq:equation_residual}
\end{equation}
and for $P\in\cP_n$ we set $\sigma_P:=\langle P,\cR_{\cS}\rangle_{\rm HS}$.
The first-order optimality conditions of the restricted problem are the
Galerkin orthogonality relations
\begin{equation} \label{eq:galerkin_ortogonality}
	\sigma_{P_j}=0,\qquad P_j\in\cS,
\end{equation}
while for $P\notin\cS$ the coefficient $\sigma_P$ is the gradient of
the full quadratic objective along the missing direction $P$.
\end{definition}

\begin{lemma}[Residual support]
\label{lem:residual_support}
Let $A_{\cS}^{(\eta)}\in\Span\cS$ and let $\cR_{\cS}$ be the equation residual
\eqref{eq:equation_residual}. Define
\begin{equation}
	\cR(\cS)
	:=
	\operatorname{supp}\bigl(-\ii\,\cL_H(G)\bigr)
	\,\cup\,
	\cS
	\,\cup\,
	\bigcup_{P_j\in\cS}\operatorname{supp}\bigl(\cL_H^2(P_j)\bigr).
	\label{eq:residual_support_set}
\end{equation}
Then $\operatorname{supp}(\cR_{\cS})\subseteq\cR(\cS)$. Consequently
$\sigma_P=0$ for every $P\notin\cR(\cS)$.
\end{lemma}

\begin{lemma}[Locality bound]
\label{lem:locality_bound}
Assume $H=\sum_{\alpha}h_\alpha Q_\alpha$, where each $Q_\alpha$ is a Pauli
string with $\wt(Q_\alpha)\le 2$ and each qubit belongs to the support of at
most $\Delta_H$ terms $Q_\alpha$. Then, for any Pauli string $P$,
$
|\operatorname{supp}(\cL_H(P))|\le \Delta_H\,\wt(P)$,
with $|\operatorname{supp}(\cL_H^2(P))|\le \Delta_H^2\,\wt(P)\bigl(\wt(P)+1\bigr)
\le \Delta_H^2\bigl(\wt(P)+1\bigr)^2$,
and, with $w_{\cS}:=\max_{P_j\in\cS}\wt(P_j)$, then
	$|\cR(\cS)|
	\;\le\;
	\bigl|\operatorname{supp}\bigl(-\ii\,\cL_H(G)\bigr)\bigr|
	\;+\;
	|\cS|\Bigl[\,1+\Delta_H^2\,(w_{\cS}+1)^2\,\Bigr]$.
The set $\cR(\cS)$ is computable by symbolic Pauli commutation in
$O(|\cR(\cS)|)$ string operations.
\end{lemma}

For the linear-ramp interpolation $H(\lambda)=(1-\lambda)H_B+\lambda H_C$ with
$H_B=-\sum_iX_i$ and $H_C$ containing single-qubit $Z_i$ fields and
$Z_iZ_k$ couplings on an interaction graph of maximum degree~$d$, each qubit
appears in at most~$d$ coupling terms plus one $Z$-field and one $X$-field
term, so the hypothesis of Lemma~\ref{lem:locality_bound} holds with
$\Delta_H\le d+2$. Hence, for bounded-weight supports, $|\cR(\cS)|$ grows
linearly in $|\cS|$ and the residual support can be enumerated explicitly,
without ever touching the $4^n$-dimensional ambient space.

\begin{remark}
Assume that, after solving the restricted problem \eqref{eq:restricted_problem} on $\cS$, the set
$\cR(\cS)$ has been computed. By the Galerkin orthogonality
\eqref{eq:galerkin_orthogonality}, $\sigma_P=0$ on $\cS$; by
Lemma~\ref{lem:residual_support}, $\sigma_P=0$ outside $\cR(\cS)$. The
\emph{pricing problem}
$
	\max_{P\in\cP_n\setminus\cS}|\sigma_P|
	=
	\max_{P\in\cR(\cS)\setminus\cS}|\sigma_P|
$ over the full $4^n$-dimensional Pauli space is therefore solved
\emph{exactly} by enumerating the finite set $\cR(\cS)\setminus\cS$,
classically and without measurements. A maximizer
	$P_\star
	\in
	\argmax_{P\in\cR(\cS)\setminus\cS}|\sigma_P|$ 
generates the enlarged support $\cS'=\cS\cup\{P_\star\}$.

\end{remark}

In practice, when working within a restricted subspace $\cS$, the exact full-space 
solution $A_\lambda^{(\eta)}$ remains unknown, making it difficult to assess the 
quality of the Galerkin approximation $A_{\cS}^{(\eta)}$. To address this, we establish 
an \emph{a posteriori} error certificate. The following Proposition~\ref{prop:residual_certificate} demonstrates 
that by monitoring the residual coefficients, more specifically, the ``leakage'' of the 
solution under the action of the Liouvillian just outside the support $\cS$, we 
can produce an upper-bound for both the functional error and the Hilbert--Schmidt 
distance to the exact regularized potential using entirely accessible, computable quantities.

\begin{proposition}[A posteriori residual certificate]
\label{prop:residual_certificate}
Let $A_\lambda^{(\eta)}$ be the exact regularised AGP solving
\eqref{eq:regularised-agp-kron} and let $A_{\cS}^{(\eta)}$ be the restricted master
solution of Definition~\ref{def:restricted_master}. Suppose
$
|\sigma_P|\le\delta$
for all $P\in\cR(\cS)\setminus\cS$,
and set $N:=|\cR(\cS)\setminus\cS|$. Then, with
$E_{\mathcal{S}}:=A_{\cS}^{(\eta)}-A_\lambda^{(\eta)}$,
\begin{itemize}
    \item[(i)] $\Phi_\eta(A_{\cS}^{(\eta)})-\Phi_\eta\bigl(A_\lambda^{(\eta)}\bigr)
	=
	\tfrac12\bigl\langle E_{\mathcal{S}},\;\cL_H^2(E_{\mathcal{S}})+\eta E_{\mathcal{S}}\bigr\rangle_{\rm HS}
	\le
	\frac{1}{2\eta}\sum_{P\in\cR(\cS)\setminus\cS}\sigma_P^2
	\le
	\frac{\delta^2N}{2\eta}$,
    \item[(ii)] $\norm{E}_{\rm HS}
	\le
	\frac{1}{\eta}\Bigl(\sum_{P\in\cR(\cS)\setminus\cS}\sigma_P^2\Bigr)^{1/2}
	\le
	\frac{\delta\sqrt N}{\eta}$.
\end{itemize}
All bounds are computable a posteriori from the residual coefficients
$\{\sigma_P\}_{P\in\cR(\cS)\setminus\cS}$ and $\eta$ alone.
\end{proposition}

The algebraic properties established in Section~\ref{sec:residual_certificate} guarantee that the regularized linear system is well-conditioned and globally solvable. However, in practical quantum simulation and many-body contexts, the full Hilbert space is often too large to handle, requiring us to project the problem onto a restricted operator support $\cS$ via a Galerkin framework. The natural question is therefore under what conditions does this subspace projection yield the exact regularized adiabatic gauge potential? The following Proposition~\ref{prop:exact_agp} shows that  if the chosen support is closed under action of the Liouvillian commutator, the projected solution matches the full-space solution identically, eliminating any truncation error.

\begin{proposition}[Exactness at commutation-closed supports]
\label{prop:exact_agp}
Suppose the support $\cS$ is commutation closed, i.e.
\begin{equation}
	\cR(\cS)\subseteq\cS .
	\label{eq:closed_support}
\end{equation}
Then $\cR_{\cS}=0$ and $A_{\cS}^{(\eta)}=A_\lambda^{(\eta)}$. 
\end{proposition}

\section{Fused Pauli--Krylov solver for the regularised AGP}
\label{sec:pauli_cg_agp}

Using the Hilbert--Schmidt inner product in~\eqref{eq:hs-inner-product} the Pauli strings~\eqref{eq:pauli_strings} form an orthonormal basis:
$
    \langle P,Q\rangle_{\rm HS}=\delta_{PQ}$, with $P,Q\in\mathcal P_n$.
Hence any Hermitian operator \(A\) can be written as $
    A=\sum_{P\in\mathcal P_n} a_P P$, for $a_P=\langle P,A\rangle_{\rm HS}$. 
Instead of storing \(A\) as a dense matrix, the fused Pauli--CG method stores
only a sparse coefficient dictionary $A
    \equiv
    \{(P,a_P):a_P\neq 0\}$.
Since the Pauli basis is orthonormal, applying CG in Pauli coordinates is
equivalent, in exact arithmetic and without truncation, to applying CG in any
other orthonormal representation of the Hilbert--Schmidt space. Thus the
Pauli formulation is not a different linear solver; it is the same CG
iteration expressed in a basis in which the commutator algebra is sparse.
{Indeed, the results in Lemma~\ref{lem:locality_bound} and Proposition~\ref{prop:residual_certificate} characterize how the dictionary increases after an application of the regularised operator, and a bound on the number of strings needed to represent the solution. In the following two subsections we make the construction practical by showing how the operator can be applied, and how this can be employed to run the Conjugate Gradient method with both a controllable memory consumption and certified error approximation.}

\subsection{Symbolic application of \texorpdfstring{\({\mathcal{B}_\eta}\)}{B eta}{: an arithmetic for Pauli expansions}} \label{sec:symbolic_application}

Assume that the Hamiltonian has a sparse Pauli expansion $
    H(\lambda)=\sum_{j=1}^{m_H} h_j Q_j,
    \qquad Q_j\in\mathcal P_n$.
Then for a Pauli string \(P \in \cP_n\), $
    \mathcal L_H(P)
    =
    [H,P]
    =
    \sum_{j=1}^{m_H} h_j[Q_j,P]$.
Each commutator is either zero or proportional to another Pauli string.
Indeed, if \(Q_j\) and \(P\) commute, then $[Q_j,P]=0$, whereas if they anticommute, then $[Q_j,P]=2Q_jP$,
up to the corresponding Pauli phase convention. Therefore
$   {\mathcal{B}_\eta} P
    =
    \mathcal L_H^2(P)+\eta P
$
can be computed by two symbolic commutator applications followed by collection
of equal Pauli strings. Indeed, for a sparse Pauli vector $p=\sum_{P\in S_p} p_P P$, 
defining $\mathcal L_H(p):=[H,p]$, we have
$    {\mathcal{B}_\eta} p
    =
    \mathcal L_H^2(p)+\eta p
    =
    [H,[H,p]]+\eta p $.
Therefore the action of \({\mathcal{B}_\eta}\) on \(p\) can be computed matrix-free by two symbolic
commutator applications:
\[
    p
    \xrightarrow{\mathcal L_H}
    [H,p]
    \xrightarrow{\mathcal L_H}
    [H,[H,p]]
    \xrightarrow{+\eta p}
    {\mathcal{B}_\eta} p .
\]
We have hence
$    [H,p]
    =
    \sum_j \sum_{P\in\mathcal S_p} h_j p_P [Q_j,P]$.
Using the properties of commutation/anticommutation explained above, it is possible to observe that each commutator maps a sparse
Pauli expansion to another sparse Pauli expansion: equal Pauli strings produced by different terms are then collected by summing their coefficients. Thus \({\mathcal{B}_\eta} p\) is obtained without forming the full matrix representation of
\({\mathcal{B}_\eta}\); only sparse dictionaries of Pauli strings and coefficients need to be manipulated.

\begin{remark}\label{rem:qtt_stuff}
This construction closely mirrors Quantized Tensor Train (QTT) methods~\cite{MR3416059}, where a $2^n \times 2^n$ matrix is reshaped into a high-dimensional tensor and compressed. Because the Pauli matrices~\eqref{eq:pauli_strings} form a complete orthogonal basis for $2 \times 2$ matrices, a local basis change on the QTT cores maps the computational basis to the Pauli basis. Consequently, the QTT representation of $H$ is mathematically equivalent to a standard Tensor Train (TT) decomposition~\cite{MR2837533} of its $n$-dimensional Pauli coefficient tensor. Exploiting this Pauli expansion format is highly advantageous for efficiently computing the required commutators.
\end{remark}

\subsection{Conjugate Gradients Method}

{By means of the arithmetic for Pauli expansions we have just described, we can show how to apply the Conjugate Gradient algorithm to compute the AGP by exploiting the sparse Pauli structure of these expansions.} Let us start from  the regularised AGP equation~\eqref{eq:regularised-agp} in {operator} form,
\begin{equation}
{\mathcal{B}_\eta\,A_\lambda^{(\eta)}
	=
	b,
	\qquad
	\mathcal{B}_\eta:=\cL_H^2+\eta I,
	\qquad
	b:=-\ii\,\cL_H(G),}
	\label{eq:krylov_regularised_agp_system}
\end{equation}
where \(G=\partial_\lambda H\){, and \(\mathcal{B}_\eta\), \(A_\lambda^{(\eta)}\), \(b\) are operators on the Hilbert--Schmidt space rather than their \(N^2\)-dimensional vectorizations}. Since \(H=H^\dagger\), the {Liouvillian \(\cL_H\)} is {self-adjoint with respect to \(\langle\cdot,\cdot\rangle_{\rm HS}\)}, hence {\(\mathcal{B}_\eta\)} is Hermitian positive definite for every \(\eta>0\). Therefore the system
\eqref{eq:krylov_regularised_agp_system} can be solved by conjugate gradient,
even though the right-hand side \(b\) is complex; see, e.g.,~\cite{MR1226007}. Starting from \({A_0}=0\), after \(k\) CG iterations one obtains ${A_k}\in \mathcal K_k({\mathcal{B}_\eta},b)$,
where $\mathcal K_k({\mathcal{B}_\eta},b)
	:=
	\operatorname{span}
	\{b,{\mathcal{B}_\eta} b,\ldots,{\mathcal{B}_\eta}^{k-1}b\}$. 
Since \({\mathcal{B}_\eta}={\cL_H^2}+\eta I\), this Krylov space coincides with
$
\mathcal K_k({\cL_H^2},b)
=
\operatorname{span}
\{b,{\cL_H^2}b,\ldots,{\cL_H^{2(k-1)}}b\}$. Consequently, the CG iterate can be written as
	${A_k}=p_{k-1}({\cL_H^2})b$
for a polynomial \(p_{k-1}\) of degree at most \(k-1\). Since the exact solution is
${A_\lambda^{(\eta)}}=({\cL_H^2}+\eta I)^{-1}b{{}=\mathcal{B}_\eta^{-1}b}
$, %
CG constructs a polynomial approximation to the scalar filter
$
s\mapsto \frac{1}{s+\eta}$.  It is important to note that CG steps discover the dominant
regularised Krylov components of the AGP.  Looking in more details at the Conjugate Gradient, starting from \(A_0=0\), the initial residual is
$
    r_0=b-{\mathcal{B}_\eta} A_0=b$,
and the first search direction is
$    p_0=r_0$.
At iteration \(k\), one computes $w_k={\mathcal{B}_\eta} p_k$, 
and chooses the CG step length
$
    \alpha_k
    =
    \frac{\langle r_k,r_k\rangle_{\rm HS}}
         {\langle p_k,w_k\rangle_{\rm HS}}$.
The AGP approximation and residual are then updated as $
    A_{k+1}=A_k+\alpha_k p_k$, and $r_{k+1}=r_k-\alpha_k w_k$.
The new search direction is $p_{k+1}=r_{k+1}+\beta_k p_k$, where
$
    \beta_k
    =
    \frac{\langle r_{k+1},r_{k+1}\rangle_{\rm HS}}
         {\langle r_k,r_k\rangle_{\rm HS}}$.
It is important to remind the reader that in the
\({\mathcal{B}_\eta}\)-energy norm, see \cite{MR1990645},
\begin{equation*}
	\|{A_\lambda^{(\eta)}}-{A_k}\|_{{\mathcal{B}_\eta}}
	\leq
	2
	\left(
	\frac{\sqrt{\cond({\mathcal{B}_\eta})}-1}
	{\sqrt{\cond({\mathcal{B}_\eta})}+1}
	\right)^k
	\|{A_\lambda^{(\eta)}}-{A_0}\|_{{\mathcal{B}_\eta}}.
\end{equation*}
Since $\cond({\mathcal{B}_\eta})
\leq
1+\frac{\|{\cL_H}\|^2}{\eta}$,
regularization directly controls the number of CG iterations needed to obtain
an accurate Krylov approximation. Moreover, in view of what was observed in Section~\ref{sec:symbolic_application}, one can easily see that the Conjugate Gradient Method for the solution of ${\mathcal{B}_\eta} {A_\lambda^{(\eta)}}=b$, can be easily implemented in Pauli coordinates using symbolic application of the operator ${\mathcal{B}_{\eta}}$. Thus the Pauli support explored by the method is generated by repeated
commutators with \(H\){, and by the result in Lemma~\ref{lem:locality_bound}, increases, under suitable hypothesis, at most polynomially together with the dimension of the adopted Krylov subspace. If convergence is not achieved quickly enough, this causes the number of dictionary terms in the expansion to grow beyond the available memory. To this end, it is necessary to introduce a Pauli expansion approximation procedure to keep memory consumption under control.} 

\subsection{Truncation as Pauli sparsification}

Without sparsification, the Pauli support may eventually become too large.
We therefore introduce a truncation operator \(T_m\), which keeps the \(m\)
largest Pauli coefficients in magnitude:
$    T_m\!\left(\sum_P a_P P\right)
    =
    \sum_{P\in S_m} a_P P$,
where \(S_m\) indexes the \(m\) largest values of \(|a_P|\).

A truncated Pauli--CG iteration takes the form
$
    A_{k+1}=T_m(A_k+\alpha_k p_k)$.
Similarly, the search direction may be sparsified:
$
    p_{k+1}=T_m(r_{k+1}+\beta_kp_k)$.
This truncation is essential for scalability, but it destroys the exact
three-term CG recurrence. In particular, the recursively updated residual
$
    r_{k+1}^{\rm rec}=r_k-\alpha_k {\mathcal{B}_\eta} p_k
$
need no longer coincide with the true residual
$
    r_{k+1}^{\rm true}=b-{\mathcal{B}_\eta} A_{k+1} 
$.
For this reason, after every truncation step we recompute the residual
exactly:
$
    r_{k+1}:=b-{\mathcal{B}_\eta} A_{k+1}
$. 
This exact residual recomputation is crucial. Loss of conjugacy merely slows
the Krylov method down, whereas residual drift can make the algorithm falsely
report convergence.

The truncated Pauli--CG method should therefore be interpreted primarily as a
support-discovery procedure. It identifies {progressively} a candidate sparse Pauli support
$
    S=\operatorname{supp}(A_k)\cup\operatorname{supp}(p_k)\cup
      \operatorname{supp}(r_k)$,
which is then refined by a Galerkin projection as explained in the next Section~\ref{sec:galerkin_refit} {, while we discuss the error introduced by this procedure in Section~\ref{sec:residual_certificate}}.

{
}
\begin{remark}
Following Remark~\ref{rem:qtt_stuff}, this truncated Krylov construction closely mirrors iterative solvers for high-dimensional systems in the Tensor Train format, such as TT-GMRES~\cite{MR3043559}. In those frameworks, operator-vector products within the Krylov loop trigger rapid rank growth, necessitating a tensor rounding step at each iteration to ensure scalability. This rank reduction breaks the exact recurrences or orthogonality of the underlying algorithm, completely analogous to the behavior of our Pauli sparsification operator $T_m$.
\end{remark}

\subsection{Galerkin refit on the discovered support}\label{sec:galerkin_refit}

Once a support \(S\subset\mathcal P_n\) has been identified, we compute the
best approximation on this support by solving the restricted Galerkin problem producing a solution $A_S$ s.t.
$
    A_S=\sum_{P\in S} c_P P$. 
The coefficients \(c_P\) are chosen so that
$
    \Pi_S {\mathcal{B}_\eta} \Pi_S c = \Pi_S b,
$
where \(\Pi_S\) denotes the orthogonal projection onto
\(\operatorname{span}(S)\). Equivalently,
$
    \langle P,{\mathcal{B}_\eta} A_S\rangle_{\rm HS}
    =
    \langle P,b\rangle_{\rm HS}$, $P\in S$.
Thus the residual $\mathcal{R}_{\mathcal{S}}=b-{\mathcal{B}_\eta} A_S$ satisfies the Galerkin orthogonality condition
\begin{equation}
\langle P,\mathcal{R}_{\mathcal{S}}\rangle_{\rm HS}=0,
\qquad P\in S,
\label{eq:galerkin_orthogonality}
\end{equation}
i.e., the Galerkin solution is the minimizer of the quadratic energy
$
    \Phi_\eta(A)
    =
    {\tfrac12\langle A,\mathcal{B}_\eta A\rangle_{\rm HS}
    -
    \langle b,A\rangle_{\rm HS}+C}
$
over \(A\in\operatorname{span}(S)\). Hence
$\displaystyle 
    A_S
    =
    \operatorname*{arg\,min}_{A\in\operatorname{span}(S)}
    \Phi_\eta(A)$.
It is therefore the \({\mathcal{B}_\eta}\)-orthogonal projection of the exact regularised
AGP \({A_\lambda^{(\eta)}}\) onto the chosen Pauli subspace.

\subsection{A posteriori residual certificate}\label{sec:residual_certificate}

The error of any sparse approximation is controlled \emph{a posteriori} by Proposition~\ref{prop:residual_certificate}. Recall that, for the restricted solution \(A_{\cS}^{(\eta)}\), its proof established the exact error representation
$
    E_{\cS}:=A_{\cS}^{(\eta)}-A_\lambda^{(\eta)}=-\mathcal{B}_\eta^{-1}\cR_{\cS}
$,
together with the bound \(\norm{E_{\cS}}_{\rm HS}\le\norm{\cR_{\cS}}_{\rm HS}/\eta\), obtained from the spectral floor \(\mathcal{B}_\eta\succeq\eta I\) of Lemma~\ref{lem:spectral_floor} and Parseval's identity in the Pauli basis.

The derivation uses nothing about \(A_{\cS}^{(\eta)}\) beyond linearity, so it applies verbatim to \emph{any} finitely supported Pauli operator \(A'\), whether or not it minimises the restricted problem. Writing its exact residual as
$
    \cR':=b-\mathcal{B}_\eta A'=-\ii\,\cL_H(G)-\mathcal{B}_\eta A'$,
and using \(\mathcal{B}_\eta A_\lambda^{(\eta)}=b\), we obtain
$
    \mathcal{B}_\eta\bigl(A'-A_\lambda^{(\eta)}\bigr)=\mathcal{B}_\eta A'-b=-\cR'$
implying $  A'-A_\lambda^{(\eta)}=-\mathcal{B}_\eta^{-1}\cR'$,
which is exactly the representation in the proof of Proposition~\ref{prop:residual_certificate} with \(A_{\cS}^{(\eta)}\) replaced by \(A'\) and \(\cR_{\cS}\) by \(\cR'\). The spectral floor of Lemma~\ref{lem:spectral_floor} then yields the computable certificate
\begin{equation} 
    \norm{A'-A_\lambda^{(\eta)}}_{\rm HS}
    \le
    \frac{\norm{\cR'}_{\rm HS}}{\eta}
    =
    \frac{1}{\eta}\Bigl(\sum_{P\in\cP_n}\bigl|\langle P,\cR'\rangle_{\rm HS}\bigr|^2\Bigr)^{1/2},
    \label{eq:agp_certificate}
\end{equation}
the Parseval expansion being justified exactly as in Proposition~\ref{prop:residual_certificate}. The sole difference is that, for a general \(A'\), the Galerkin orthogonality \eqref{eq:galerkin_ortogonality} need not hold, so the sum ranges over all of \(\operatorname{supp}(\cR')\) instead of the finite set \(\cR(\cS)\setminus\cS\) of Lemma~\ref{lem:residual_support}. The bound is independent of how \(A'\) was produced and survives arbitrary truncation, loss of conjugacy, or heuristic support updates, provided the residual \(\cR'=b-\mathcal{B}_\eta A'\) is recomputed exactly.

When \(A'\) is the Galerkin refit \(A_{\cS}^{(\eta)}\), the orthogonality \eqref{eq:galerkin_ortogonality} restores \(\sigma_P=0\) on \(\cS\) and, by Lemma~\ref{lem:residual_support}, \(\sigma_P=0\) outside \(\cR(\cS)\); the certificate \eqref{eq:agp_certificate} then collapses to Proposition~\ref{prop:residual_certificate}(ii),
\[
    \norm{A_{\cS}^{(\eta)}-A_\lambda^{(\eta)}}_{\rm HS}
    \le
    \frac{1}{\eta}\Bigl(\sum_{P\in\cR(\cS)\setminus\cS}\sigma_P^2\Bigr)^{1/2}
    =
    \frac{\norm{b-\mathcal{B}_\eta A_{\cS}^{(\eta)}}_{\rm HS}}{\eta}.
\]
Thus the residual norm gives a computable upper bound on the error of the sparse AGP approximation: if the certificate is small, the discovered Pauli support is sufficient; if it remains large even after adding more Pauli strings, the AGP is not well approximated within the imposed sparsity budget.

Algorithm~\ref{alg:truncated-pauli-cg} describes the implementation used in the numerical experiments.
The method is a Pauli-coordinate version of the regularised AGP equation.  A fixed number of
matrix-free conjugate-gradient iterations is first used to discover a candidate Pauli support.
The support is then truncated by coefficient magnitude, $T_{m,w_{\max}}
\left(
\sum_{P\in\mathcal P_n} a_P P
\right)
=
\sum_{P\in S_{m,w_{\max}}} a_P P$,
where $S_{m,w_{\max}}$ contains the $m$ largest coefficients $|a_P|$ among the Pauli strings
satisfying ${\rm wt}(P)\le w_{\max}$; refitted by solving the restricted
Galerkin problem, and certified by recomputing the exact symbolic residual

\begin{algorithm}[htb!]
\caption{Truncated Pauli--CG iteration for the regularised AGP}
\label{alg:truncated-pauli-cg}
\begin{algorithmic}[1]
\Require Hamiltonian $H(\lambda)$, derivative $G=\partial_\lambda H(\lambda)=H_C-H_B$,
regularization $\eta>0$, number of CG iterations $K$, optional truncation budget $m$,
optional maximum Pauli weight $w_{\max}$.
\Ensure A Pauli--CG approximation $A^{(\eta), K}_\lambda$ to the regularised AGP $A^{(\eta)}_\lambda$.
\State Initialize
\[
A_0=0,
\qquad
r_0=b,
\qquad
p_0=r_0 .
\]
\For{$k=0,\ldots,K-1$}
    \State $w_k={\mathcal{B}_\eta} p_k.$ \Comment{Apply the matrix-free Pauli operator}
    \State $\displaystyle \alpha_k
    =
    \frac{\langle r_k,r_k\rangle_{\rm HS}}
         {\langle p_k,w_k\rangle_{\rm HS}}$ .
    \State $    \widetilde A_{k+1}
    =
    A_k+\alpha_k p_k .
    $ \Comment{Take the {candidate} CG step}
    \If{truncation is active}
        \State $A_{k+1}
        =
        T_{m,w_{\max}}
        \left(\widetilde A_{k+1}\right).$ \Comment{Truncate the Pauli expansion}
    \Else
        \State $A_{k+1}=\widetilde A_{k+1}$.
    \EndIf
    \State 
    $
    r_{k+1}
    =
    b-{\mathcal{B}_\eta} A_{k+1}.
    $ \Comment{Recompute the residual exactly}
    \State 
    $
    \displaystyle \beta_k
    =
    \frac{\langle r_{k+1},r_{k+1}\rangle_{\rm HS}}
         {\langle r_k,r_k\rangle_{\rm HS}} .
    $
    \State 
    $
    \widetilde p_{k+1}
    =
    r_{k+1}+\beta_k p_k .
    $ \Comment{Form the new CG search direction}
    \If{truncation is active}
        \State 
        $p_{k+1}
        =
        T_{m,w_{\max}}
        \left(\widetilde p_{k+1}\right).$ \Comment{Truncate the search direction:}
    \Else
        \State $p_{k+1}=\widetilde p_{k+1}$.
    \EndIf
\EndFor

\State \Return $A^{(\eta), K}_\lambda$.

\end{algorithmic}
\end{algorithm}
The residual is recomputed as
$r_{k+1}=b-{\mathcal{B}_\eta} A_{k+1}$ rather than updated recursively as
$r_{k+1}^{\rm rec}=r_k-\alpha_k {\mathcal{B}_\eta} p_k$.
This distinction matters once truncation is applied, because truncation destroys the exact
three-term CG recurrence and the recursively updated residual may no longer coincide with
the true residual.

We emphasize that, after truncation, Algorithm \ref{alg:truncated-pauli-cg} does no longer coincides with an exact conjugate-gradient
method but rather with an inexact version, see \cite{MR2058070,MR1740397}.  Therefore the classical CG monotonicity and finite-step convergence guarantees do
not directly apply to the truncated iterates.  The role of the truncated iteration is instead
support discovery.  Reliability is recovered a posteriori by recomputing the true residual
$
    r_{k+1}^{\rm true}
$
and by applying the residual certificate \eqref{eq:agp_certificate} after the
Galerkin refit.

\subsection{{Magnitude truncation, Galerkin refit and Hermitianity}}
Every Hermitian operator admits a real Pauli
expansion
$   A=\sum_{P\in\mathcal P_n} a_P P$, $a_P=\langle P,A\rangle_{\rm HS}\in\mathbb R$.
After \(K\) conjugate-gradient iterations in Pauli coordinates, we obtain an approximation
$
    A_K=\sum_{P\in\mathcal P_n} a_P^{(K)}P$.
Its Pauli support typically grows with \(K\).  As explained in the previous section, to keep the counterdiabatic generator
implementable, we combine magnitude truncation with a Galerkin refit on the retained support. 
Given a support budget \(m\), and optionally a maximum Pauli weight \(w_{\max}\), define
\[
    \mathcal S
    =
    \left\{
    P\in\mathcal P_n:
    |a_P^{(K)}| \text{ is among the } m \text{ largest retained coefficients and }
    \operatorname{wt}(P)\leq w_{\max}
    \right\}.
\]
The truncated operator is
\begin{equation}
    T_{m,w_{\max}}(A_K)
    =
    \sum_{P\in\mathcal S} a_P^{(K)}P .
    \label{eq:pauli-truncation}
\end{equation}
    
In practice, this requires only a partial sort of the current sparse Pauli dictionary by
\(|a_P^{(K)}|\).  Coefficients below numerical precision are discarded, and the cost scales with
the number of stored Pauli strings rather than with the full dimension \(4^n\).

Direct truncation does not, in general, give the best approximation on the retained support,
because removing Pauli directions changes the optimal coefficients on the remaining ones.
We therefore recompute the coefficients on $V_{\mathcal S}:=\operatorname{span}\{P:\,P\in\mathcal S\}$ by solving classically the Galerkin projection of the regularised AGP equation
$    A_{\mathcal S}^{(\eta)}
    =
    \sum_{P\in\mathcal S} c_P P$, and $M_{\mathcal S}c=f_{\mathcal S}$,
where $(M_{\mathcal S})_{PQ}
    =
    \langle P,\mathcal B_\eta(Q)\rangle_{\rm HS}$,
    $(f_{\mathcal S})_P
    =
    \langle P,b\rangle_{\rm HS}$, and $P,Q\in\mathcal S$.
The reduced matrix is assembled by applying \(\mathcal B_\eta\) symbolically to each retained
Pauli string \(Q\in\mathcal S\), collecting only the components in \(\mathcal S\).  Since
\(|\mathcal S|=m\) is small compared with \(4^n\), the reduced system is then solved directly.
Thus truncation selects the implementable Pauli support, while the Galerkin refit computes the
best regularised-AGP coefficients on that support.

The Pauli-coordinate representation also makes Hermiticity transparent.  Since \(H\) and
\(G\) are Hermitian, the right-hand side $b=-\ii[H,G]$
is Hermitian.  Moreover, \(\mathcal B_\eta=\mathcal L_H^2+\eta I\) maps Hermitian operators
to Hermitian operators and is self-adjoint positive definite on the Hilbert--Schmidt space.
Therefore the exact solution \(A_\lambda^{(\eta)}\) is Hermitian and has real Pauli
coefficients.

The same property is preserved by the numerical construction.  The CG iterates are represented
with real coefficients in the Hermitian Pauli basis.  The truncation
\eqref{eq:pauli-truncation} is simply an orthogonal projection onto a coordinate subspace of
this real Hermitian basis, and therefore removes Pauli terms without introducing any
anti-Hermitian component.  Finally, the Galerkin matrix \(M_{\mathcal S}\) is a real symmetric
positive-definite matrix, and the Galerkin solution \(c=M_{\mathcal S}^{-1}f_{\mathcal S}\) is
real.  Consequently, $A_{\mathcal S}^{(\eta)}
    =
    \sum_{P\in\mathcal S} c_P P$ 
is Hermitian by construction.  No post-processing symmetrization is required, and the
counterdiabatic propagator $\exp(-\ii \alpha A_{\mathcal S}^{(\eta)})$ is unitary by construction.

\section{QAOA implementation}
\label{sec:qaoa_implementation}

This section specifies the circuit model used in the computational experiments.  We
allow the problem Hamiltonian to be a generic Hermitian Pauli Hamiltonian
    $H_C
    =
    \sum_{\mu\in\mathcal I_C} c_\mu P_\mu$, for $c_\mu\in\mathbb R$,
where each \(P_\mu\) is an \(n\)-qubit Pauli string, i.e., $P_\mu\in\{I,X,Y,Z\}^{\otimes n}$.
In the classical Ising or QUBO case, \(H_C\) contains only \(Z\)-type strings
and is therefore diagonal in the computational basis, but this is not fully necessary for the following developments. The mixer convention used throughout the experiments is
    $H_B=-\sum_{i=1}^n X_i$,
and   $U_B(\beta)=\exp(-\ii\beta H_B)
    =
    \exp\!\left(+\ii\beta\sum_{i=1}^n X_i\right)$,
with initial state $|+\rangle^{\otimes n}$. Thus, both the cost layer and the mixer layer are treated as Hamiltonian
evolution operators. The interpolation used to define the AGP is
$
    H(\lambda)
    =
    (1-\lambda)H_B+\lambda H_C$, with $G=\partial_\lambda H(\lambda)=H_C-H_B$.
For depth $p$, the implemented Linear Ramp-QAOA (LR-QAOA) state is
$    |\psi_p^{\rm LR}\rangle
    =
    \prod_{i=0}^{p-1}
    U_B(\beta_i)U_C(\gamma_i)
    |+\rangle^{\otimes n}$,
where the product is ordered by increasing layer index as applied to the
state.  Equivalently, the \(i\)-th layer first applies the cost evolution
\(U_C(\gamma_i)\) and then the mixer evolution \(U_B(\beta_i)\).  The CD
extension (denoted by LR-CD-QAOA in the following) inserts an AGP-generated unitary after the cost and mixer operations
of the same ramp layer
$    |\psi_p^{\rm CD}\rangle
    =
    \prod_{i=0}^{p-1}
    \exp(-\ii\alpha_i A_i)
    U_B(\beta_i)U_C(\gamma_i)
    |+\rangle^{\otimes n}$.
Here \(A_i\) denotes the regularised adiabatic gauge potential associated with
the interpolation point \(\lambda_i\). 

Moreover, at depth \(p\), the used linear-ramp angles are
    $\beta_i=\left(1-\frac{i}{p}\right)\Delta_\beta$,
    $\gamma_i=\frac{i+1}{p}\Delta_\gamma$, for $i=0,\ldots,p-1$.
The CD layer uses the endpoint schedule
$\lambda_i=\frac{i+1}{p}$, $\alpha_i=\Delta_{\rm CD}\lambda_i(1-\lambda_i)$.
Therefore, for \(p=1\), the only interpolation point is \(\lambda_0=1\), and
hence \(\alpha_0=0\).  With this particular endpoint schedule, the CD layer is
the identity at depth one, so CD-LR-QAOA coincides exactly with LR-QAOA for
the same values of \((\Delta_\beta,\Delta_\gamma)\).  This identity is a
consequence of the chosen endpoint discretization and is not a generic
property of all possible CD-QAOA schedules. In the experiments, \(A_i\) may be the
exact spectral regularised AGP, a fixed-\(K\) conjugate-gradient iterate for
the regularised AGP equation, or a Pauli-compressed approximation of that
iterate.  It is important to note that we use a variational version of LR-QAOA, see \cite{Dehn2025}, which uses only two variational parameters
$
    (\Delta_\beta,\Delta_\gamma),
$
i.e., the minimized function is
  $  \min_{\Delta_\beta,\Delta_\gamma}
    \;
    \Phi_{\rm LR}(\Delta_\beta,\Delta_\gamma)
    :=
    \langle
        \psi_p^{\rm LR}(\Delta_\beta,\Delta_\gamma)
        |
        H_C
        |
        \psi_p^{\rm LR}(\Delta_\beta,\Delta_\gamma)
    \rangle$.
Anaolougoiusly, in the case of CD-LR-QAOA, we propose the three scalar variational version where the parameters are optimized in the
CD ramp, $(\Delta_\beta,\Delta_\gamma,\Delta_{\rm CD})$. The corresponding variational problem is
$    \min_{\Delta_\beta,\Delta_\gamma,\Delta_{\rm CD}}
    \;
    \Phi_{\rm CD}(\Delta_\beta,\Delta_\gamma,\Delta_{\rm CD})
    :=
    \langle
        \psi_p^{\rm CD}(\Delta_\beta,\Delta_\gamma,\Delta_{\rm CD})
        |
        H_C
        |
        \psi_p^{\rm CD}(\Delta_\beta,\Delta_\gamma,\Delta_{\rm CD})
    \rangle$,
In both cases we use the box constraints 
$\Delta_\beta\in [-\pi,\pi]$, $\Delta_\gamma\in [-\pi,\pi]$, and $\Delta_{\rm CD}\in[-\pi,\pi]$.

\section{Numerical experiments}
\label{sec:numerics}

In this section we validate both the algorithmic and the theoretical framework we have discussed. Specifically, in Section~\ref{sec:theory_validation} we validate the theoretical framework with respect to the effect of the regularization parameter $\eta$, the energy resolution mechanism, and the usage of the (truncated) CG method for matrices and vectors expressed in Pauli's basis on the Ferromagnetic Chain (FC) Ising problem. In Section~\ref{subsec:asc-maxcut-generation} we discuss instances obtained by perturbing FC with other components. All numerical experiments were implemented in \texttt{Python 3.10.20} using \texttt{Qiskit 2.4.1} with \texttt{Aer 0.17.2} and executed on a MacBook Pro
with Apple Silicon processor \({\rm Apple \; M3  \; Pro}\), \({\rm 18 \; GB \; RAM}\) memory, running
macOS \({\rm Thaoe \;  26.5.1 }\).  The code used to reproduce the experiments is publicly
available on the GitHub repository \href{https://github.com/StefanoCipolla/Inexact_CG_CD_LR_QAOA}{StefanoCipolla/Inexact\_CG\_CD\_LR\_QAOA}.

\subsection{Theory Validation}\label{sec:theory_validation}
The first set of experiments aims to validate the developed theory. Accordingly, all experiments in the first part of this section use exact, dense statevector simulation in complex double precision: circuits are executed on \texttt{Qiskit Aer's} statevector simulator, and all reported quantities are evaluated as exact expectation values through the \texttt{EstimatorV2} primitive at zero target precision (i.e., without sampling/shot noise). Throughout this section we consider the model problem: 

\begin{equation}
    H(\lambda)= \lambda \underbrace{(-J\sum_i Z_iZ_{i+1})}_{H_C} + (1-\lambda) \underbrace{( -\sum_iX_i)}_{H_B},
    \label{eq:tfim_hamiltonian}
\end{equation}
see \cite[Sec. B]{wbbs-s8fs}.
For each pair \((n,p)\), the LR-QAOA ramp parameters are optimized first, and
the same LR solution is then used to initialize the CD-LR optimizations.  The
CD-LR ansatz optimizes the additional counterdiabatic amplitude
\(\Delta_{\rm CD}\), while keeping the same endpoint schedule
$\lambda_i=\frac{i+1}{p}$, with $\alpha_i=\Delta_{\rm CD}\lambda_i(1-\lambda_i)$.
Since CD-LR-QAOA reduces exactly to LR-QAOA when
\(\Delta_{\rm CD}=0\), this initialization embeds the LR solution inside the
larger CD variational family.  The CD optimization is then performed over $(\Delta_\beta,\Delta_\gamma,\Delta_{\rm CD})$,
so that any improvement over LR can be attributed to the additional counterdiabatic degree of freedom rather than to a different initialization of the ramp parameters.

\subsubsection{Regularised AGP Theory validation}

\begin{figure}[t]
    \centering
    \includegraphics[width=0.45\linewidth]{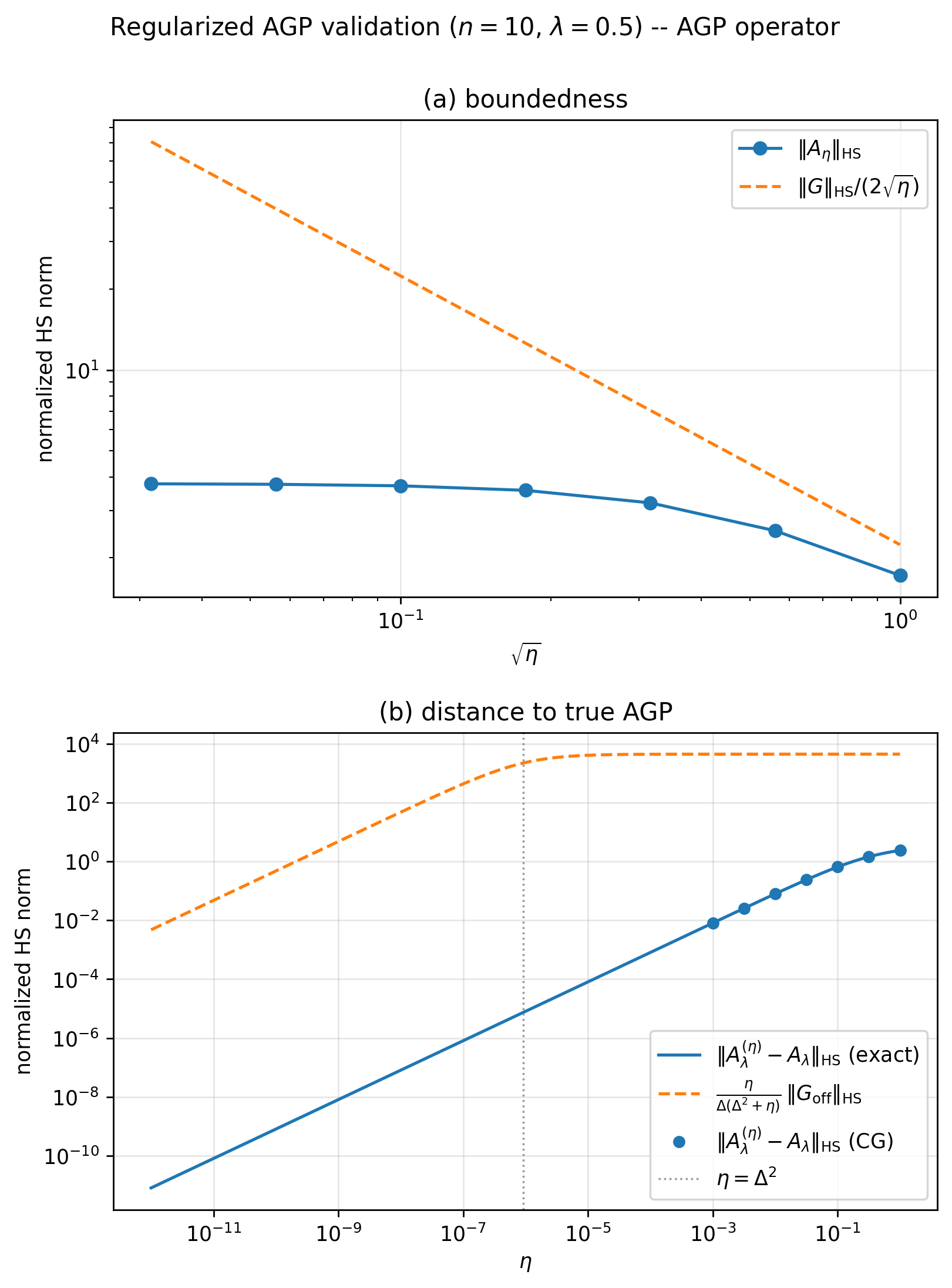}
    \includegraphics[width=0.45\linewidth]{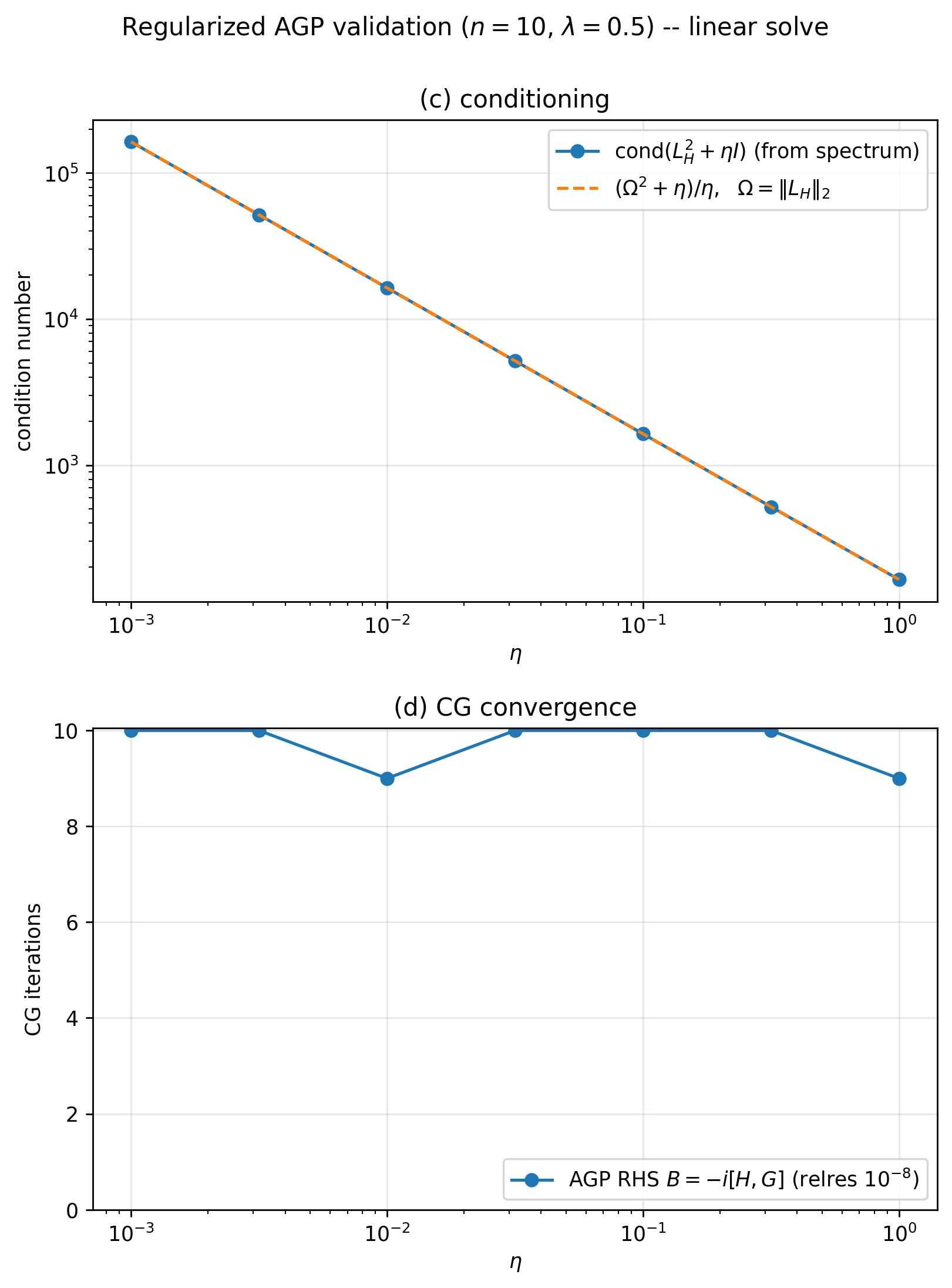}
    \caption{Matrix-free validation of the regularised AGP theory for the
    ordered-Ising interpolation \eqref{eq:tfim_hamiltonian} with \(n=10\).  All quantities are computed in
    the computational basis without diagonalizing \(H(\lambda)\).  Panel (a)
    verifies the boundedness of the regularised AGP as a function of the
    regularization scale \(\sqrt{\eta}\).  Panel (b) reports the deviation from
    the unregularised AGP, showing that the regularised solution approaches the
    exact AGP as \(\eta\) decreases.  Panel (c) compares the condition number of
    the regularised squared Liouvillian with the theoretical upper bound.
    Panel (d) reports the number of matrix-free conjugate-gradient iterations
    required to solve the regularised AGP equation to relative residual
    \(10^{-8}\).}
    \label{fig:agp_theory_validation}
\end{figure}

Figure~\ref{fig:agp_theory_validation} validates the main theoretical
properties of the regularised AGP on the ordered-Ising interpolation \eqref{eq:tfim_hamiltonian}.  The
regularised AGP is computed by solving $\left(\mathcal L_H^2+\eta I\right)A_\eta
    =
    -i\mathcal L_H(G)$, and $\mathcal L_H(X)=[H,X]$,
using both, the spectral formula obtained after
diagonalizing \(H(\lambda)\) and a matrix-free conjugate-gradient method.

Panel (a) confirms the boundedness effect of the regularization.  The
normalized Hilbert--Schmidt norm of \(A_\eta\) remains well below the
theoretical bound in Proposition \ref{prop:boundedness}.  This behaviour is
consistent with the regularised approach 
which suppresses small-frequency components and prevents the AGP from becoming
singular near small gaps.

Panel (b) illustrates the regularization bias.  As \(\eta\) is reduced, the
regularised solution approaches the unregularised AGP, while larger values of
\(\eta\) increasingly attenuate the AGP components.  This confirms the expected
trade-off: small \(\eta\) gives a more faithful approximation of the exact AGP,
whereas larger \(\eta\) produces a smoother and more strongly regularised
generator.

Panel (c) validates the conditioning estimate for the regularised squared
Liouvillian.  The condition number of
$
    \mathcal L_H^2+\eta I
$
scales inversely with \(\eta\), and closely follows the theoretical upper bound.
This confirms that the regularization shifts the spectrum away from zero and
turns the AGP equation into a well-posed positive-definite linear system.

Finally, Panel (d) reports the corresponding matrix-free CG iteration counts.
Despite the growth of the worst-case condition number for small \(\eta\), the
number of iterations remains modest for the AGP right-hand side
$
    -i[H,G].
$
This indicates, somehow, that for this particular Hamiltonian, the solution of the AGP equation lies in a small-dimensional Krylov space.

\subsubsection{Energy-resolution mechanism}

The Hamiltonian \eqref{eq:tfim_hamiltonian} naturally realises the spectral structure targeted by the regularised AGP. As shown in Figure \ref{fig:tfim_gap_scaling}, in finite size, the two lower states form a nearly degenerate doublet with splitting \(\Delta_{\rm split}=E_1-E_0\), while the gap to the first excitation outside the doublet, \(\Delta_{\rm out}=E_2-E_1\), remains \(O(1)\) away.  Figure~\ref{fig:tfim_gap_scaling} confirms this separation for periodic chains up to \(n=14\).  The regularised filter in Lemma~\ref{lem:filter-formula} is therefore useful in the window
\begin{equation}
    \Delta_{\rm split}\ll \sqrt\eta \ll \Delta_{\rm out}.
    \label{eq:tfim_eta_window}
\end{equation}
In this regime the CD generator need not resolve the exponentially small doublet splitting, but it can still suppress transitions out of the ordered manifold.
\begin{figure}[t]
    \centering
    \includegraphics[width=0.95\linewidth]{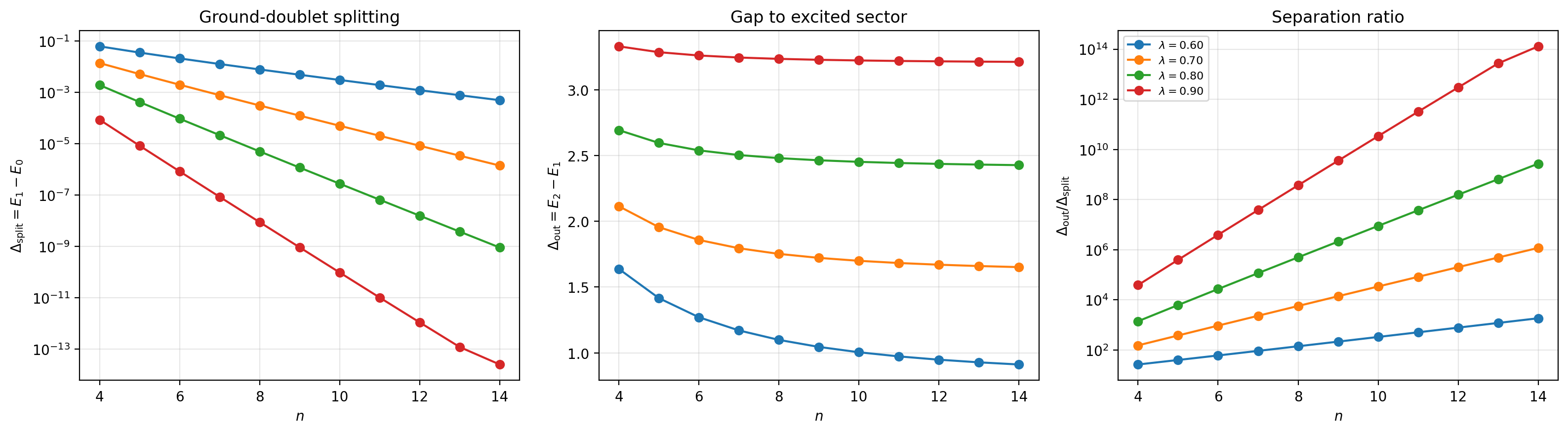}
    \caption{ Spectral structure of $H(\lambda)$.  The ground-doublet splitting decays rapidly with system size, while the outer excitation gap remains order one.  This creates a broad window in which \(\sqrt\eta\) can suppress doublet-resolving components without washing out the gap separating the ordered manifold from excited states.}
    \label{fig:tfim_gap_scaling}
\end{figure} 
For the Ferromagnetic Chain (FC) endpoint
\begin{equation} \label{eq:hchain}
H_C=-J\sum_{i=1}^n Z_iZ_{i+1},
\qquad J>0,    
\end{equation}
with periodic boundary conditions, every computational-basis state
\(|z\rangle=|z_1\cdots z_n\rangle\) is an eigenvector.  Since $Z_i|z\rangle=(-1)^{z_i}|z\rangle$,
we have
$
H_C|z\rangle
=
-J\sum_{i=1}^n (-1)^{z_i+z_{i+1}}|z\rangle
=
E(z)|z\rangle$, 
where \(z_{n+1}=z_1\).  Each term contributes \(-J\) when
\(z_i=z_{i+1}\) and \(+J\) when \(z_i\neq z_{i+1}\).  Hence the energy is
minimized when all neighbouring bits agree.  With periodic boundary
conditions this gives exactly two minimizing bitstrings, $
00\cdots0
\text{ and }
11\cdots1.
$
Both have energy $E_{\min}=-nJ$. Thus the endpoint has a two-dimensional classical ground-state manifold
spanned by $|0\cdots0\rangle$, and $|1\cdots1\rangle$.
Therefore the natural success
metric is the probability that the final state lies in this ordered manifold,
$
P_{\rm ord}(\psi)
=
|\langle 0\cdots0|\psi\rangle|^2
+
|\langle 1\cdots1|\psi\rangle|^2$,
hence, measuring the capability of the regularised AGP to resolve the exponentially small splitting inside the ordered doublet and suppress transitions out of the ordered manifold.

\begin{figure}[t]
    \centering
    \includegraphics[width=.49\linewidth]{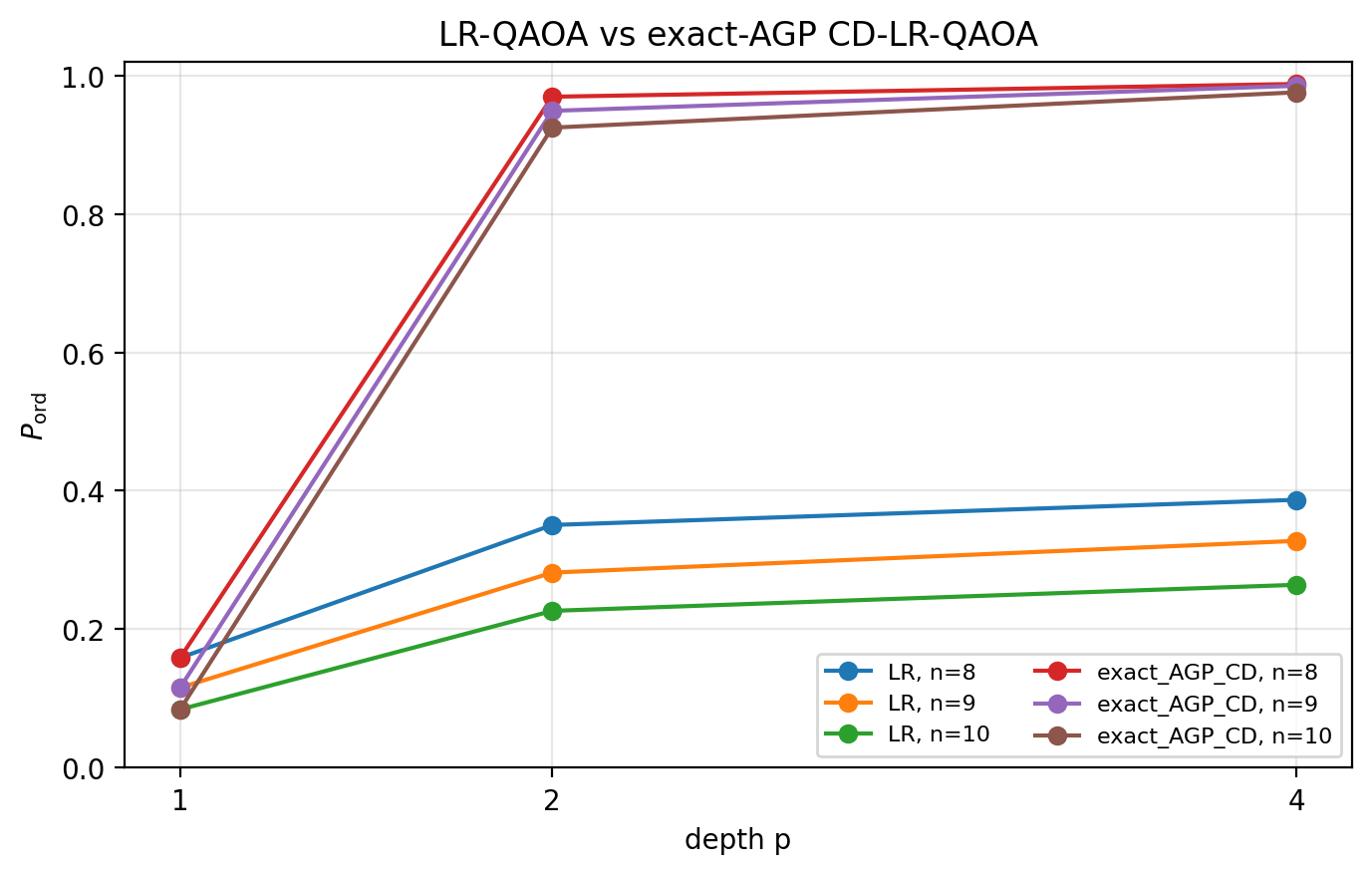}
    \includegraphics[width=.49\linewidth]{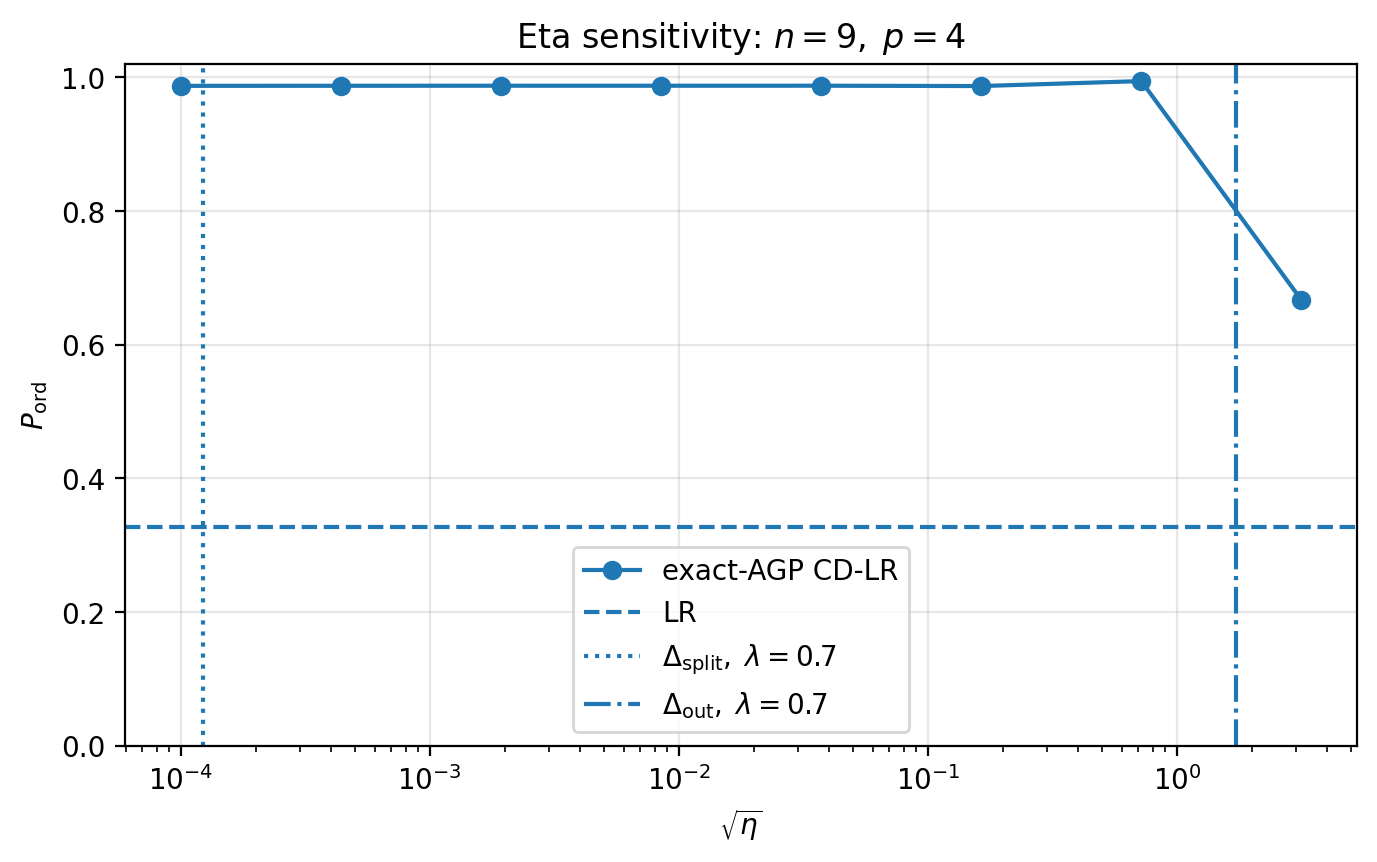}
    \caption{Theory-validation experiments for the ferromagnetic chain endpoint.
    Left: ordered-manifold probability \(P_{\rm ord}\) as a function of the
    circuit depth \(p\), comparing LR-QAOA with exact-AGP CD-LR-QAOA for
    \(n=8,9,10\). Right: sensitivity of exact-AGP CD-LR-QAOA to the
    regularization scale \(\sqrt{\eta}\) for \(n=9\) and \(p=2\). The dashed
    horizontal line gives the LR-QAOA baseline, while the vertical lines mark
    the reference spectral scales \(\Delta_{\rm split}\) and
    \(\Delta_{\rm out}\) of \(H(\lambda)\).}
    \label{fig:Pord_and_sensitivity}
\end{figure}

The left panel of Figure~\ref{fig:Pord_and_sensitivity} reports $P_{\rm ord}(\psi)$
obtained by LR-QAOA and exact-AGP CD-LR-QAOA. 
The results show a clear separation between the two ansatz classes.  Standard
LR-QAOA improves only mildly with the depth over the range \(p=1,2,4\), and its
performance deteriorates as \(n\) increases.  In contrast, the exact-AGP
CD-LR ansatz reaches a much larger ordered-manifold probability already at
\(p=2\), with \(P_{\rm ord}\) close to \(0.9\) for all tested system sizes.
At \(p=4\), the exact-AGP CD-LR curves approach \(P_{\rm ord}\simeq 1\),
indicating that the regularised counterdiabatic layer strongly suppresses
leakage out of the ordered subspace.  The coincidence between LR-QAOA and
CD-LR-QAOA at \(p=1\) is consistent with the endpoint CD schedule: for
\(p=1\), the only CD interpolation point is \(\lambda_0=1\), hence
$   \alpha_0
    =
    \Delta_{\rm CD}\lambda_0(1-\lambda_0)
    =
    0$, 
and the CD layer is the identity. The right panel of Figure~\ref{fig:Pord_and_sensitivity} studies the dependence
on the regularization scale \(\sqrt{\eta}\) for \(n=9\) and \(p=2\).  For each
value of \(\eta\), the regularised AGP is recomputed and the CD-LR parameters
are reoptimized, while the LR baseline is kept fixed.  The plot shows that the
CD improvement is robust over a broad interval of regularization values:
\(P_{\rm ord}\) remains close to \(0.9\) for small and moderate values of
\(\sqrt{\eta}\).  When \(\sqrt{\eta}\) becomes comparable with, or larger than,
the gap to the excited sector, the performance decreases.  This is consistent
with the spectral-filter interpretation of the regularised AGP, see Lemma \ref{lem:filter-formula}. Therefore transitions with energy differences \(|\omega|\ll\sqrt{\eta}\) are
suppressed, whereas transitions with \(|\omega|\gg\sqrt{\eta}\) are retained. In the regime 
$    \Delta_{\rm split}
    \ll
    \sqrt{\eta}
    \ll
    \Delta_{\rm out}$,
the AGP does not resolve the exponentially small splitting inside the
ordered doublet, but still captures transitions from the ordered manifold to the excited sector.

\subsubsection{Conjugate Gradient Theory Validation}

\begin{figure}[t]
    \centering
    \includegraphics[width=.515\linewidth]{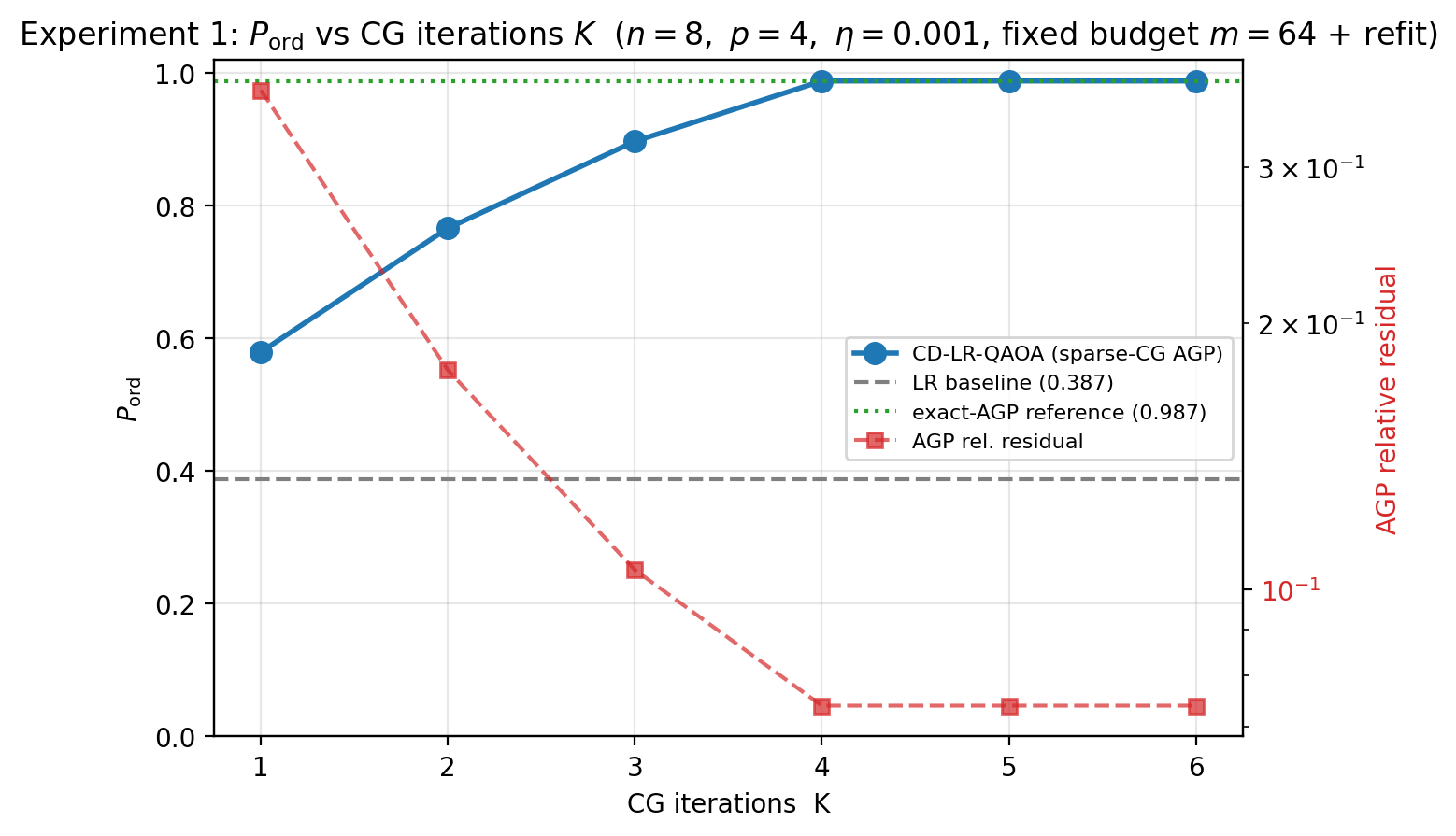}
    \includegraphics[width=.47\linewidth]{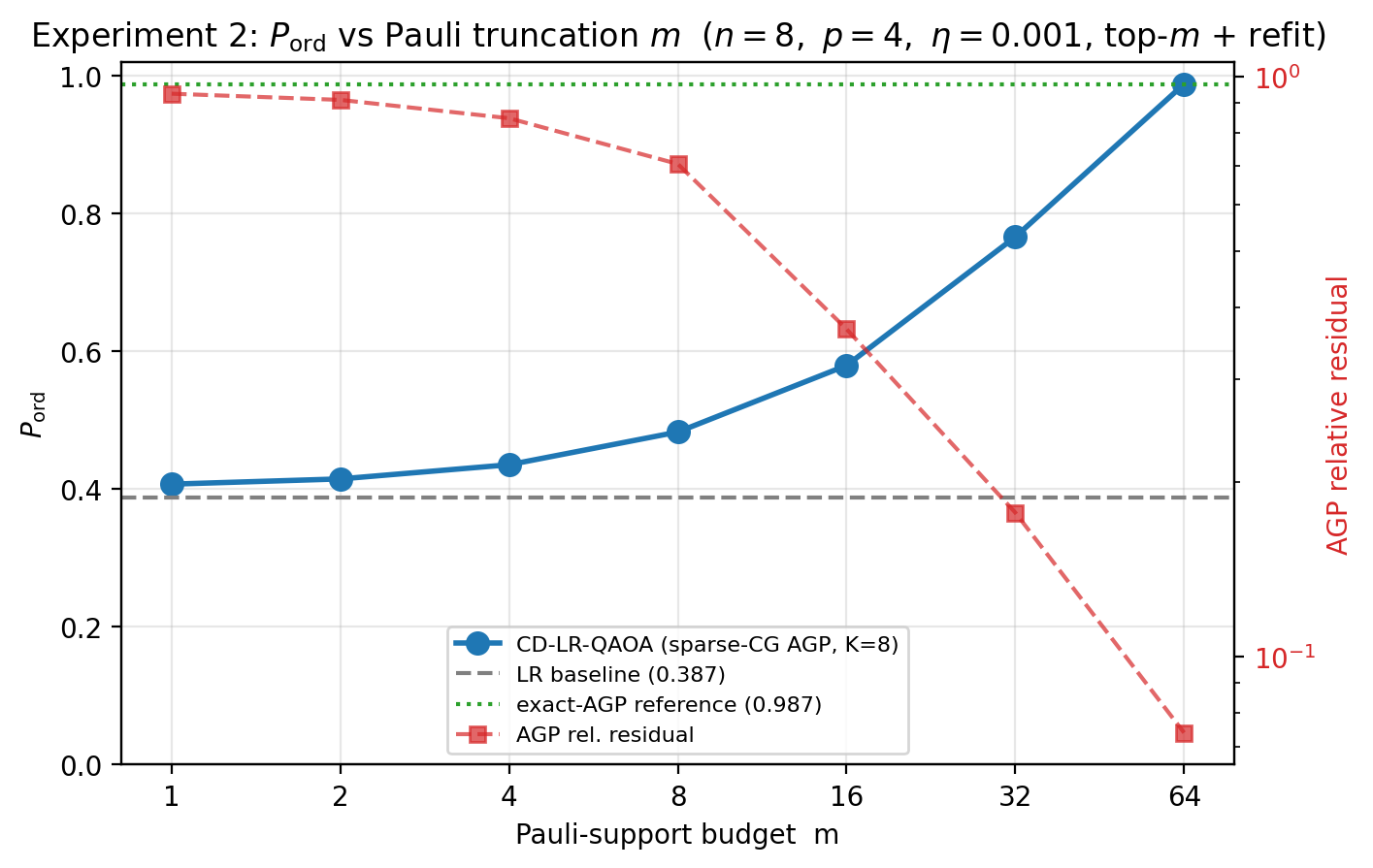}
    \caption{Convergence of the sparse conjugate-gradient (CG) construction of the regularized
    adiabatic gauge potential (AGP) on the ferromagnetic  chain ($n=8$ qubits, depth $p=4$,
    regularization $\eta=10^{-3}$, $\lambda=1/2$). Both panels report the ordered-state probability
    $P_{\mathrm{ord}}$ (left axis, blue, linear) and the AGP relative residual
    $\rho=\lVert(\mathcal{L}_H^2+\eta I)A-b\rVert_F/\lVert b\rVert_F$, $b=-i[H,G]$ (right axis,
    red, logarithmic), against the LR-QAOA baseline ($P_{\mathrm{ord}}=0.387$) and the exact
    dense-AGP reference ($P_{\mathrm{ord}}=0.987$). Left: dependence on the number of CG iterations
    $K$ at fixed support budget $m=64$. Right: dependence on the Pauli-support budget $m$ at fixed,
    converged $K=8$. In both cases the retained support is followed by an exact Galerkin refit.}
    \label{fig:K_m_dependencey}
\end{figure}

In Figure~\ref{fig:K_m_dependencey} we report the behaviour of the sparse CG construction of the
regularized AGP ($\lambda=1/2$) as a function of the number of conjugate-gradient iterations $K$ at a fixed
Pauli-support budget $m=64$ (left panel) and as a function of the Pauli-support budget $m$ at a
fixed, converged iteration count $K=8$ (right panel); in both cases the retained support is
followed by an exact Galerkin refit, and we track the ordered-state probability
$P_{\mathrm{ord}}$ together with the AGP relative residual. The computational results
presented in the figures highlight how the two parameters bound the approximation along
complementary axes: $K$ sets the depth of the iterative solve, while $m$ sets the sparsity of the
resulting operator. As either parameter is increased, the AGP reltive residual  decreases monotonically
and $P_{\mathrm{ord}}$ rises from the LR-QAOA baseline to the exact dense-AGP reference, so that
the exact AGP is recovered only when both the iteration count and the support budget are
sufficient. In the left panel a single iteration ($K=1$) already yields $P_{\mathrm{ord}}\approx
0.58$, and the exact-AGP reference is essentially attained by $K=4$, beyond which the AGP relative residual plateaus
at the floor imposed by the $m=64$ truncation rather than by the iteration; in the right panel
$P_{\mathrm{ord}}$ grows from $\approx0.41$ at $m=1$ to the reference value at $m=64$, with the
sharp improvement occurring once the dominant Pauli terms enter the support. The computational
results also highlight how $P_{\mathrm{ord}}$ saturates well before the AGP relative residual is driven to zero,
indicating that a modest iterative budget ($K\approx4$) together with a moderate support already
captures essentially all of the counterdiabatic benefit -- a favourable property in view of
scaling the construction to larger systems.

\subsection{Blended  instances}
\label{subsec:asc-maxcut-generation}

In order to validate the superiority of our proposal on \textit{difficult instances} for LR-QAOA, we consider in this section blended instances obtained by perturbing the ferromagnetic chain in  \eqref{eq:hchain} with other Hamiltonians.  The considered cost Hamiltonians are hence
$H_C(\epsilon)
    =
    \epsilon H_{\rm FC}
    +
    (1-\epsilon)H$, for $\epsilon\in[0,1]$.
The term \(H_{\rm FC}\) is the same as used in the
previous experiments, but with heavy/light weight couplings.  In this section, we perturb such \(H_{\rm FC}\)  with $H$ being i) a weighted MaxCut Hamiltonian generated from a random
graph or, ii) the Hamiltonian obtained form the QUBO instance \texttt{ms\_03\_050\_002} from \cite{Koch2026} (Market Split Collection).
For this experiment, each counterdiabatic generator was computed
using the regularised AGP equation
at the interpolation points \(\lambda_i=i/p\).  We used a regularisation parameter
$
    \eta=10^{-3},
    \hbox{ and }
    K=30
$
truncated Pauli-space conjugate-gradient iterations.  During the CG iteration, the
intermediate Pauli dictionaries were truncated to at most \(300\) Pauli strings.  After CG,
we retained the \(m=200\) largest Pauli coefficients, with no Pauli-weight cap, and then
performed an exact Galerkin refit on the selected support.  
It is important to note that the support size \(m\) should be interpreted only as a proxy for implementation cost.  

Figure~\ref{fig:maxcut} reports the behaviour of the approximation ratio, computed as in \cite[Eq. (6)]{montanezbarrera2025evaluatingperformancequantumprocessing}, of LR-QAOA and LR-CD-QAOA on two
twenty-qubit  perturbed MaxCut/MarketSplit instances with \(\epsilon\) close to one.  This regime is
particularly informative because the cost Hamiltonian is dominated by the FC component, and
therefore inherits a structured low-energy landscape with small spectral separations.  The
standard LR-QAOA ramp improves as the depth increases, but it is not able to reach a high
approximation ratio: even at \(p=8\), its performance remains clearly below the optimum.

The regularised AGP correction substantially changes this behaviour.  For both
\(\epsilon=0.90\) and \(\epsilon=0.95\), LR-CD-QAOA gives a uniformly better approximation
ratio at every tested depth. {The improvement is especially significant because the CD layer is
computed from the same interpolation Hamiltonian and is designed to suppress diabatic
transitions.  When \(\epsilon\) is close to one, the FC
structure makes these diabatic interactions particularly detrimental for LR-QAOA.  The AGP
correction mitigates this effect by adding a counterdiabatic generator that reduces leakage away
from the relevant low-energy sector.  As a result, the CD-enhanced ramp reaches ratios close to
the exact optimum, while the uncorrected LR ramp remains noticeably suboptimal.
}
\begin{figure}[t]
    \centering
    \includegraphics[width=.49\linewidth]{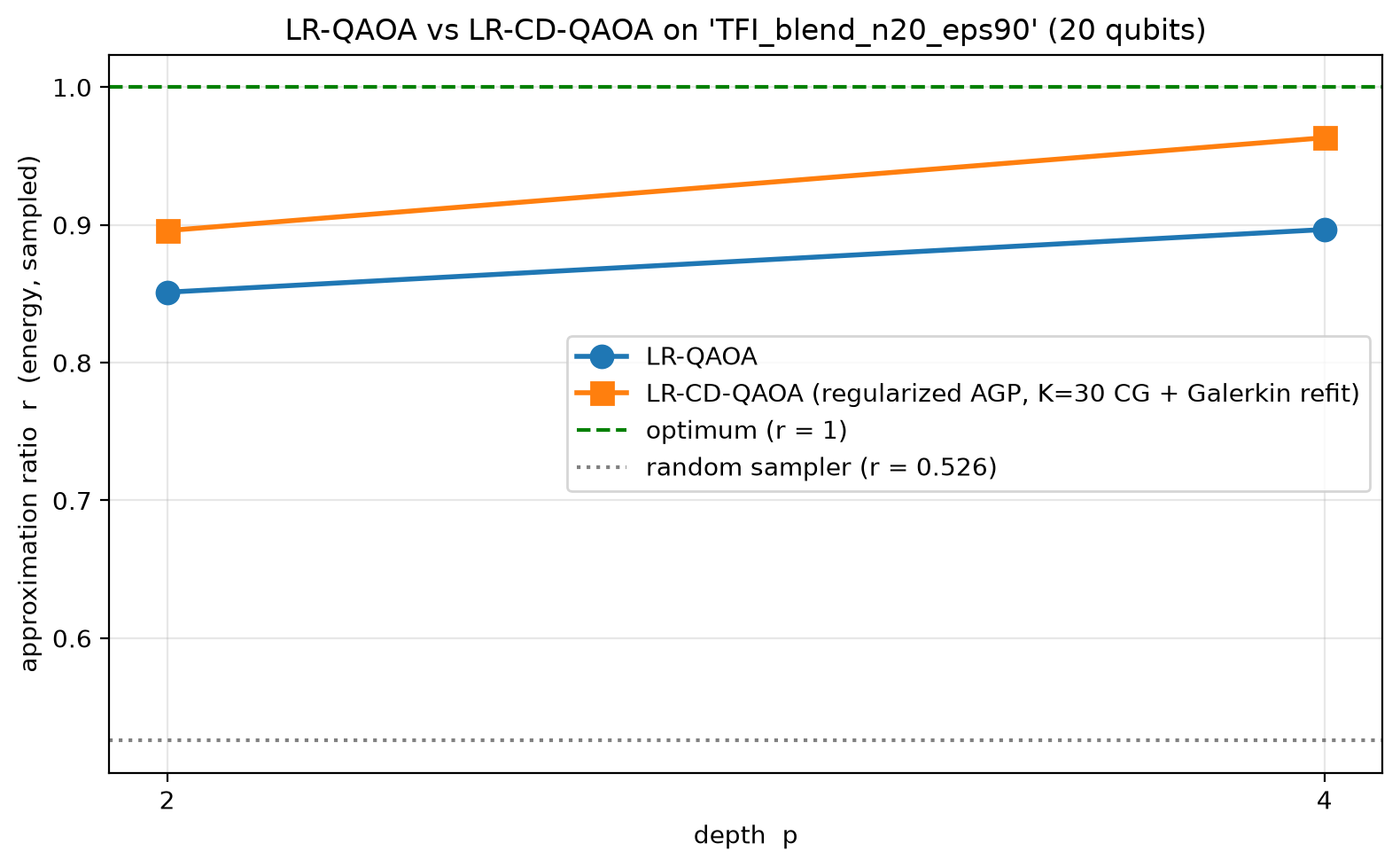}
    \includegraphics[width=.49\linewidth]{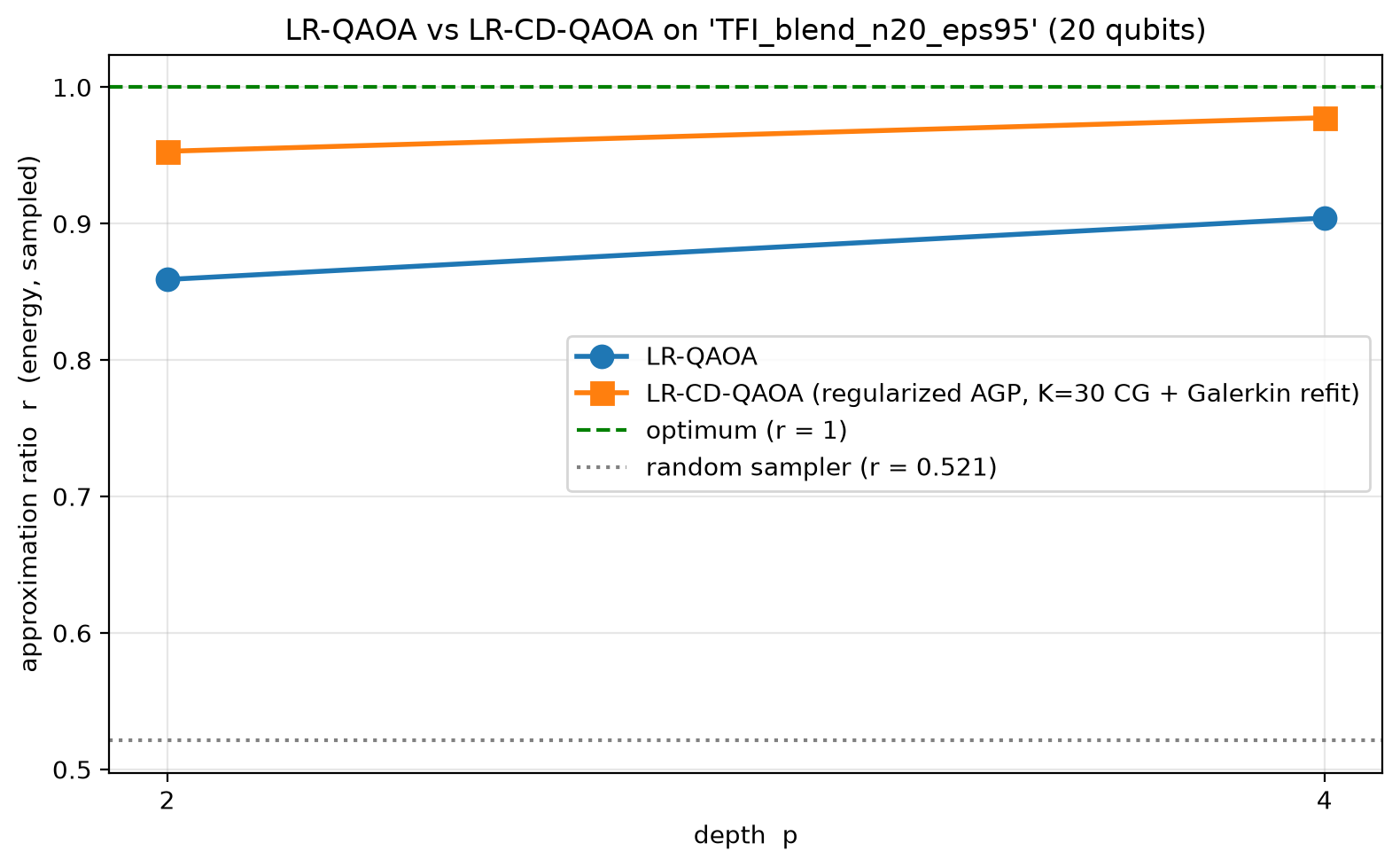} \\
    \includegraphics[width=.49\linewidth]{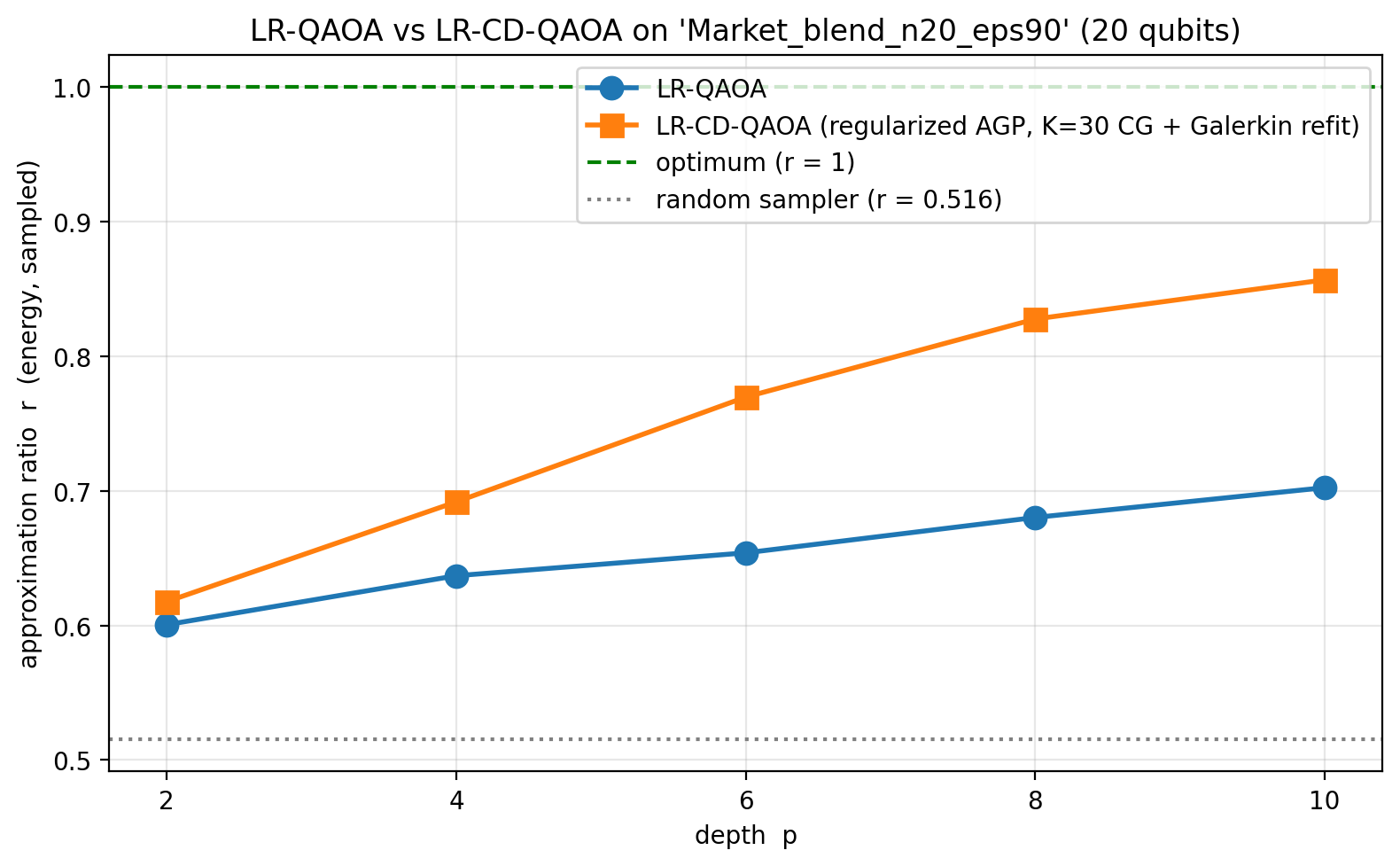}
    \includegraphics[width=.49\linewidth]{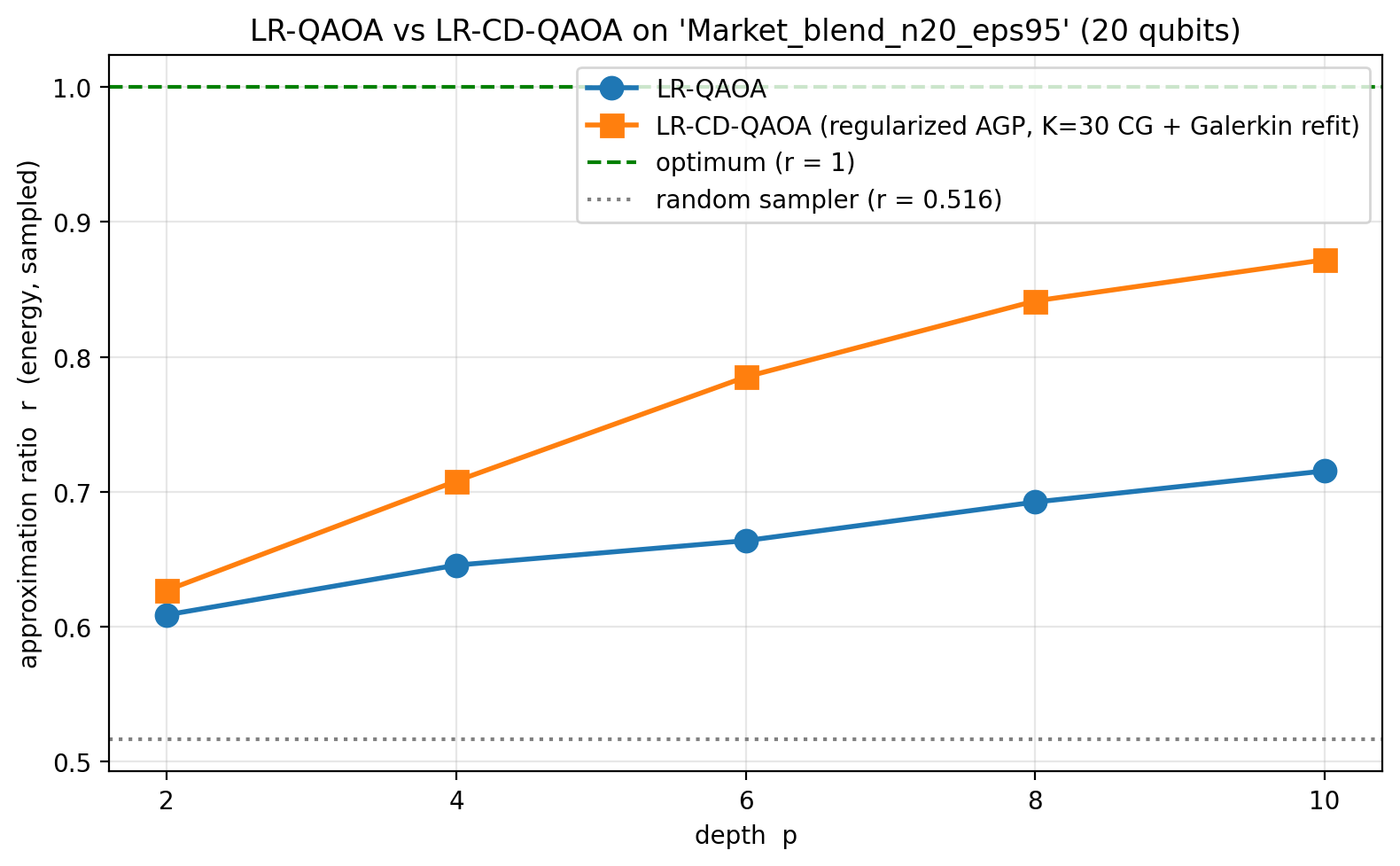} \\
    \caption{
    Approximation ratio on the FC perturbed MaxCut (upper panel) and MarketSplit (lower panel)  instances with \(n=20\) qubits.
    Left: \(\epsilon=0.90\). Right: \(\epsilon=0.95\).
    The parameter \(\epsilon\) controls the contribution of the FC component in
    \(H_C(\epsilon)=\epsilon H_{\rm FC}+(1-\epsilon)H\), so both instances are close to
    the structured FC regime while retaining a random weighted-MaxCut perturbation.
    The dashed green line denotes the optimal value \(r=1\), while the dotted grey line denotes
    the random-sampling baseline.  LR-QAOA improves with the depth \(p\), but remains visibly
    separated from the optimum.  In contrast, LR-CD-QAOA, using the regularised AGP computed
    with \(K=30\) conjugate-gradient iterations followed by a Galerkin refit, consistently achieves
    higher ratios and approaches the optimal value already at moderate depth.}
    \label{fig:maxcut}
\end{figure}

\section{Conclusions}
\label{sec:conclusions}
We have introduced a Pauli-sparse counterdiabatic extension of linear-ramp QAOA based on an inexact CG solution of the regularized AGP equation. Rather than implementing the AGP as a dense operator, our approach treats it as a computational object from which a sparse set of useful Pauli rotations can be extracted. The framework integrates five key components: a regularized AGP equation, a matrix-free Pauli-coordinate Krylov solver, on-the-fly truncation of the Pauli expansion, a Galerkin refit on the discovered support, and an a posteriori residual certificate. This method is motivated by optimization Hamiltonians where exponentially small spectral gaps make exact single-eigenstate tracking impractical. For quantum optimization, a more realistic goal is preserving probability mass within a low-energy subspace. The regularized AGP inherently supports this: the parameter $\eta$ acts as an energy-resolution scale that suppresses transitions below $\sqrt{\eta}$ while retaining those across larger spectral separations. This allows the counterdiabatic correction to ignore microscopic splittings inside the target manifold while reducing leakage to higher excited states. Algorithmically, the truncated Pauli-CG method offers a practical construction route. Because the action of $\mathcal{L}_H^2 + \eta I$ is evaluated via symbolic Pauli commutators, the dense Liouvillian matrix is never explicitly formed. Truncation during the CG iterations transforms the linear solve into a support-discovery procedure that optimizes for a finite gate budget. The subsequent Galerkin refit then optimizes the coefficients on this sparse support, while the residual certificate quantifies the remaining approximation error. Numerical experiments validate this mechanism. On Ferromagnetic Chain (FC) instances, the regularized counterdiabatic layer substantially increases the success probability within the low-energy manifold without resolving the exponentially small ground-state splitting. Experiments with sparse CG approximations demonstrate that a small number of Krylov iterations and a moderate Pauli-support budget capture most of the dense AGP's benefits. Finally, on difficult QUBO instances, the truncated Pauli-CG correction achieves significantly higher approximation ratios than plain LR-QAOA, confirming that these improvements survive both Pauli truncation and the Galerkin refit. Several directions remain open. Building on the connection highlighted in Remark~\ref{rem:qtt_stuff}, a natural next step is reformulating this framework within the TT/QTT formalism to directly compare memory footprints and leverage recent randomized tensor-rounding techniques~\cite{MR4968834}. Additionally, implementing the Pauli sequence arithmetic within a highly scalable, distributed-memory software package represents an important avenue for future engineering.

\ACKNOWLEDGMENT{During the preparation of this manuscript, the authors used LLMs  to assist with language editing and improving the clarity of the exposition. All AI-assisted material was reviewed and revised by the authors, who take full responsibility for the mathematical content and the final manuscript.}

\bibliographystyle{references} %
\bibliography{refs} %

\begin{APPENDIX}{Proofs of Theoretical Results} \label{appendix}

This appendix contains the detailed proofs for all the results established in the paper, which are placed here due to their technical nature and length to ensure a more concise presentation of our main contributions.

\begin{proof}{Proof of Lemma~\ref{lem:filter-formula}}
Let $E_{ab}=|a\rangle\langle b|$.  In column-stacking notation, $\vecop(E_{ab})=   \overline b  \otimes a$.  Since $H|a\rangle=E_a|a\rangle$, we have $(I_N\otimes H)\vecop(E_{ab})=E_a\vecop(E_{ab})$, and $(H^T\otimes I_N)\vecop(E_{ab})=E_b\vecop(E_{ab})$. Therefore $L_H\vecop(E_{ab})=(E_a-E_b)\vecop(E_{ab}) =\omega_{ab}\vecop(E_{ab})$, and $L_H^2\vecop(E_{ab})=\omega_{ab}^2\vecop(E_{ab})$. Equivalently, $\cL_H(E_{ab})=\omega_{ab}E_{ab}$, and $\cL_H^2(E_{ab})=\omega_{ab}^2E_{ab}$. Taking the $(a,b)$ entry of \eqref{eq:unregularised-agp}, or equivalently the component along $\vecop(E_{ab})$ in \eqref{eq:unregularised-agp-kron}, gives $\omega_{ab}^2(A_\lambda)_{ab} =-\ii\omega_{ab}G_{ab}$. If $a\neq b$ and $\omega_{ab}\neq0$, this yields $(A_\lambda)_{ab}=-\frac{\ii G_{ab}}{\omega_{ab}}$. Similarly, the component along $\vecop(E_{ab})$ in \eqref{eq:regularised-agp-kron} gives $(\omega_{ab}^2+\eta)(A_\lambda^{(\eta)})_{ab} =-\ii\omega_{ab}G_{ab}$, which proves the second formula.  Dividing by the unregularised entry proves the filter relation.
\end{proof}

\begin{proof}{Proof of Proposition~\ref{prop:distance}}
By Lemma~\ref{lem:filter-formula}, for $a\neq b$,
$
    (A_\lambda^{(\eta)}-A_\lambda)_{ab}
    =-\ii G_{ab}\left(\frac{\omega_{ab}}{\omega_{ab}^2+\eta}-\frac1{\omega_{ab}}\right)
    =\ii G_{ab}\frac{\eta}{\omega_{ab}(\omega_{ab}^2+\eta)}.
$
Squaring and summing over $a\neq b$ gives~\eqref{eq:distance-exact}.  If
$|\omega_{ab}|\geq\Delta$, then
$
    \frac{\eta}{|\omega_{ab}|(\omega_{ab}^2+\eta)}
    \leq
    \frac{\eta}{\Delta(\Delta^2+\eta)}$.
The bound follows immediately using the fact that $\| \cdot \|_{HS}$ is unitarily invariant.
\end{proof}

\begin{proof}{Proof of Proposition~\ref{prop:boundedness}}
The scalar function
$g_\eta(\omega)=\frac{|\omega|}{\omega^2+\eta}$ is maximized at $|\omega|=\sqrt\eta$, with maximum value $1/(2\sqrt\eta)$.  From \eqref{eq:filter-formula}, $
    |(A_\lambda^{(\eta)})_{ab}|
    \leq \frac{1}{2\sqrt\eta}|G_{ab}|$. 
Squaring and summing gives the result.
\end{proof}

\begin{proof}{Proof of Proposition~\ref{prop:conditioning}}
Vectorization is a linear isomorphism, so the spectrum and conditioning of the
operator equation can be studied through the ordinary matrix
$L_H^2+\eta I_{N^2}$.  Since $H=H^\dagger$, both $I_N\otimes H$ and
$H^T\otimes I_N$ are Hermitian matrices.  Hence $L_H=L_H^\dagger$. Consequently,
$L_H^2=L_H^\dagger L_H\succeq0$. Therefore all eigenvalues of $L_H^2+\eta I_{N^2}$ are at least $\eta>0$, and
the matrix is Hermitian positive definite. We now identify its eigenvalues.  Let $E_{ab}=|a\rangle\langle b|$.  As shown
in the proof of Lemma~\ref{lem:filter-formula}, $L_H\vecop(E_{ab})=\omega_{ab}\vecop(E_{ab})$. Thus $(L_H^2+\eta I_{N^2})\vecop(E_{ab})
    =(\omega_{ab}^2+\eta)\vecop(E_{ab})$.
Hence the eigenvalues of the regularised coefficient matrix are exactly
$   \omega_{ab}^2+\eta,
    \qquad a,b=1,\ldots,N$.
The smallest full-space eigenvalue is at least $\eta$.  The largest eigenvalue
is at most $\norm{L_H}_2^2+\eta=\Omega^2+\eta$.  This gives
\eqref{eq:full-conditioning}.  On an off-diagonal gauge-fixed subspace for
which $|\omega_{ab}|\geq\Delta$, the smallest eigenvalue is at least
$\Delta^2+\eta$, while the same upper bound holds.  This proves
\eqref{eq:gauge-fixed-conditioning}.

Finally, for every matrix $X$,
$
    \norm{\cL_H(X)}_{HS}=\norm{HX-XH}_{HS}
    \leq \norm{HX}_{HS}+\norm{XH}_{HS}
    \leq 2\norm{H}_2\norm{X}_{HS}$. 
Therefore $\Omega=\norm{\cL_H}_2=\norm{L_H}_2\leq2\norm{H}_2$.  Substitution
into \eqref{eq:full-conditioning} gives \eqref{eq:conditioning-h-bound}.
\end{proof}

\begin{proof}{Proof of Lemma~\ref{lem:spectral_floor}}
By Proposition~\ref{prop:conditioning} the eigenvalues of $B_\eta$ are exactly
$\omega_{ab}^2+\eta$, $a,b=1,\dots,N$, so the spectrum lies in
$[\eta,\Omega^2+\eta]$ and the lower edge is attained at $\omega_{aa}=0$.
Restriction to a subspace can only raise the smallest eigenvalue { by the Courant-Fischer Theorem~\cite[Theorem 4.2.11]{HornJohnson}}.
\end{proof}

\begin{proof}{Proof of Lemma~\ref{lem:residual_support}}
Writing $A_{\cS}^{(\eta)}=\sum_{P_j\in\cS}a_jP_j$,
\[
\cR_{\cS}
=
-\ii\,\cL_H(G)
-\sum_{P_j\in\cS}a_j\,\cL_H^2(P_j)
-\eta\sum_{P_j\in\cS}a_jP_j ,
\]
so every Pauli string carrying a nonzero coefficient of $\cR_{\cS}$ appears in
$-\ii\,\cL_H(G)$, in one of the expansions $\cL_H^2(P_j)$, or in
$\{P_j:P_j\in\cS\}$ itself (see the term $-\eta A_{\cS}^{(\eta)}$). Since the Pauli strings are
orthonormal under $\langle\cdot,\cdot\rangle_{\rm HS}$, the coefficient of any
$P\notin\cR(\cS)$ vanishes, hence $\sigma_P=0$.
\end{proof}

\begin{proof}{Proof of Lemma~\ref{lem:locality_bound}}
Let \(P\) be a Pauli string and write
$
H=\sum_{\alpha}h_{\alpha}Q_{\alpha}$, and
$\mathcal L_H(P)=[H,P]
=\sum_{\alpha}h_{\alpha}[Q_{\alpha},P]$.
A term \([Q_{\alpha},P]\) is nonzero only if \(Q_{\alpha}\) and \(P\)
anticommute. In particular, this can happen only if
$
\operatorname{supp}(Q_{\alpha})\cap \operatorname{supp}(P)\neq \emptyset$.
By assumption, each qubit belongs to the support of at most \(\Delta_H\)
Hamiltonian terms \(Q_{\alpha}\). Since \(P\) has \(\operatorname{wt}(P)\)
non-identity sites, there are therefore at most $\Delta_H \operatorname{wt}(P)$
Hamiltonian terms whose support intersects \(\operatorname{supp}(P)\).
Hence at most \(\Delta_H\operatorname{wt}(P)\) commutators
\([Q_{\alpha},P]\) can be nonzero. Consequently, $
|\operatorname{supp}(\mathcal L_H(P))|
\le
\Delta_H \operatorname{wt}(P)$.

Moreover, each nonzero commutator of two Pauli strings is again a scalar
multiple of a single Pauli string. Since $\operatorname{wt}(Q_{\alpha})\le 2$, 
and \(Q_{\alpha}\) must overlap \(\operatorname{supp}(P)\) in at least one
qubit in order to fail to commute with \(P\), multiplying \(P\) by
\(Q_{\alpha}\) can introduce at most one new non-identity site. Therefore
each Pauli string appearing in \([Q_{\alpha},P]\) has weight at most $\operatorname{wt}(P)+1$.

We now apply the same argument once more. Every Pauli string \(P'\)
appearing in \(\mathcal L_H(P)\) satisfies
$
\operatorname{wt}(P')\le \operatorname{wt}(P)+1$.
For each such \(P'\), the previous bound gives
$|\operatorname{supp}(\mathcal L_H(P'))|
\le
\Delta_H \operatorname{wt}(P')
\le
\Delta_H(\operatorname{wt}(P)+1)$. 
Since \(\mathcal L_H(P)\) contains at most
\(\Delta_H\operatorname{wt}(P)\) Pauli strings, the total number of Pauli
strings that can appear after applying \(\mathcal L_H\) a second time is
bounded by
$
|\operatorname{supp}(\mathcal L_H^2(P))|
\le
\Delta_H\operatorname{wt}(P)\,
\Delta_H(\operatorname{wt}(P)+1)$.
Thus $
|\operatorname{supp}(\mathcal L_H^2(P))|
\le
\Delta_H^2\operatorname{wt}(P)(\operatorname{wt}(P)+1)
\le
\Delta_H^2(\operatorname{wt}(P)+1)^2$.

It remains to bound the size of the residual index set \(\mathcal R(S)\).
Using the definition of \(\mathcal R(S)\) from \eqref{eq:residual_support_set}, the residual can
only contain Pauli strings coming from three sources:
$
\operatorname{supp}(-i\mathcal L_H(G))$,
$S$, and $\bigcup_{P_j\in S}\operatorname{supp}(\mathcal L_H^2(P_j))$. Therefore, by the union bound,
\[
|\mathcal R(S)|
\le
|\operatorname{supp}(-i\mathcal L_H(G))|
+
|S|
+
\sum_{P_j\in S}
|\operatorname{supp}(\mathcal L_H^2(P_j))|.
\]
Let $
w_S:=\max_{P_j\in S}\operatorname{wt}(P_j)$. Then, for every \(P_j\in S\),
$
|\operatorname{supp}(\mathcal L_H^2(P_j))|
\le
\Delta_H^2(w_S+1)^2$.
Substituting this into the previous estimate gives
$
|\mathcal R(S)|
\le
|\operatorname{supp}(-i\mathcal L_H(G))|
+
|S|
+
|S|\Delta_H^2(w_S+1)^2$.
Equivalently,
$
|\mathcal R(S)|
\le
|\operatorname{supp}(-i\mathcal L_H(G))|
+
|S|\left[1+\Delta_H^2(w_S+1)^2\right]$.

Finally, the set \(\mathcal R(S)\) is computable by symbolic Pauli
commutation. Indeed, for each \(P_j\in S\), one computes the nonzero
commutators \([Q_{\alpha},P_j]\), then the nonzero second commutators
\([Q_{\beta},[Q_{\alpha},P_j]]\). Each commutator produces either zero or
one Pauli string up to a scalar coefficient. Hence constructing
\(\mathcal R(S)\) requires only enumerating the strings that actually
appear, giving an implementation cost proportional to the number of
generated strings, namely \(O(|\mathcal R(S)|)\) symbolic string
operations, up to constants depending on the representation of Pauli
strings.
\end{proof}

\begin{proof}{Proof of Proposition~\ref{prop:residual_certificate}}
Since $A_\lambda^{(\eta)}$ solves the regularised AGP equation, we have
$$\cR_{\cS}=(\cL_H^2+\eta I)\bigl(A_\lambda^{(\eta)}-A_{\cS}^{(\eta)}\bigr)=-(\cL_H^2+\eta I)(E) \Rightarrow E=-(\cL_H^2+\eta I)^{-1}(\cR_{\cS}).$$
For any $A$, we have $A = A_{\lambda}^{(\eta)}+A - A_{\lambda}^{(\eta)}$ with $E := A - A_{\lambda}^{(\eta)}$. Using the quadratic form of the objective, we have
\[
\Phi_\eta(A_{\lambda}^{(\eta)}+E)
=
\frac12
\left\langle A_{\lambda}^{(\eta)}+E,\,
\mathcal{B}_\eta(A_{\lambda}^{(\eta)}+E)
\right\rangle_{\rm HS}
{+
\left\langle \ii \cL_{H}(G) ,A_{\lambda}^{(\eta)}+E\right\rangle_{\rm HS}}
+
C,
\]
where \(C\) is independent of \(A\). Expanding the first term gives
\[
\begin{aligned}
\frac12
\left\langle A_{\lambda}^{(\eta)}+E,\,
\mathcal{B}_\eta(A_{\lambda}^{(\eta)}+E)
\right\rangle_{\rm HS}
&=
\frac12
\left\langle A_{\lambda}^{(\eta)},\mathcal{B}_\eta A_{\lambda}^{(\eta)}\right\rangle_{\rm HS}
+
\left\langle E,\mathcal{B}_\eta A_\eta\right\rangle_{\rm HS}
+
\frac12
\left\langle E,\mathcal{B}_\eta E\right\rangle_{\rm HS}.
\end{aligned}
\]
Therefore,
\[
\begin{aligned}
\Phi_\eta(A_{\lambda}^{(\eta)}+E)
&=
\Phi_\eta(A_{\lambda}^{(\eta)})
+
\left\langle E,\mathcal{B}_\eta A_{\lambda}^{(\eta)}\right\rangle_{\rm HS}
{+
\left\langle \ii \cL_{H}(G),E\right\rangle_{\rm HS}}
+
\frac12
\left\langle E,\mathcal{B}_\eta E\right\rangle_{\rm HS}.
\end{aligned}
\]
The middle two terms combine as
\[
\left\langle E,\mathcal{B}_\eta A_{\lambda}^{(\eta)}\right\rangle_{\rm HS}
{+}
\left\langle \ii \cL_{H}(G),E\right\rangle_{\rm HS}
=
\left\langle E,\mathcal{B}_\eta A_{\lambda}^{(\eta)}b{+}\ii \cL_{H}(G)\right\rangle_{\rm HS}.
\]
Since \(A_{\lambda}^{(\eta)}\) solves the full regularised AGP equation, the linear term vanishes, i.e.,
\[
\left\langle E,\mathcal{B}_\eta A_{\lambda}^{(\eta)}b{+}\ii \cL_{H}(G)\right\rangle_{\rm HS}=0.
\]
Consequently, only the pure quadratic error term remains:
\[
\underbrace{\Phi_\eta(A_{\lambda}^{(\eta)}+E)}_{=\Phi_\eta(A)}-\Phi_\eta(A_{\lambda}^{(\eta)})
=
\frac12
\left\langle E,b{\mathcal{B}_\eta} E\right\rangle_{\rm HS}.
\]
Finally, considersing $A=A_{\mathcal{S}}^{(\eta)}$ and using $E=-(\cL_H^2+\eta I)^{-1}(\cR_{\cS})$,
\[
\Phi_\eta(A_{\cS}^{(\eta)})-\Phi_\eta\bigl(A_\lambda^{(\eta)}\bigr)
=
\tfrac12\bigl\langle (\cL_H^2+\eta I)^{-1}(\cR_{\cS}),\,\cR_{\cS}\bigr\rangle_{\rm HS}
\le
\frac{\norm{\cR_{\cS}}_{\rm HS}^2}{2\eta},
\]
by the Lemma~\ref{lem:spectral_floor}. Analogously, we have $\norm{E}_{HS}\le\frac{\norm{\cR_{\cS}}_{HS}}{\eta}$. It remains to express the residual norm in terms of its Pauli coefficients.
Since the Pauli strings form an orthonormal basis with respect to the
normalised Hilbert--Schmidt inner product, Parseval's identity gives
$
\|\mathcal{R}_{\mathcal{S}}\|_{\rm HS}^2
=
\sum_{P\in\mathcal P_n}
\left|
\langle P,\mathcal{R}_{\mathcal{S}}\rangle_{\rm HS}
\right|^2$.
By definition, $\sigma_P:=\langle P,\mathcal{R}_{\mathcal{S}}\rangle_{\rm HS}$. Hence
$
\|\mathcal{R}_{\mathcal{S}}\|_{\rm HS}^2
=
\sum_{P\in\mathcal P_n}
|\sigma_P|^2$.
By
Lemma~\ref{lem:residual_support}, the Galerkin orthogonality
\eqref{eq:galerkin_ortogonality}, we have (i)--(ii).
\end{proof}

\begin{proof}{Proof of Proposition~\ref{prop:exact_agp}}
By Lemma~\ref{lem:residual_support} and \eqref{eq:closed_support},
$\operatorname{supp}(\cR_{\cS})\subseteq\cR(\cS)\subseteq\cS$, while the
Galerkin orthogonality \eqref{eq:galerkin_orthogonality} gives
$\sigma_{P_j}=0$ for every $P_j\in\cS$. Hence every Pauli coefficient of
$\cR_{\cS}$ vanishes, i.e.\ $\cR_{\cS}=0$, and $A_{\cS}^{(\eta)}$ solves the full
regularised AGP equation \eqref{eq:regularised-agp-kron}. Since
$B_\eta\succ 0$ by Lemma~\ref{lem:spectral_floor}, the solution is unique,
so $A_{\cS}^{(\eta)}=A_\lambda^{(\eta)}$.
\end{proof}

\end{APPENDIX}

\end{document}